\iffalse 
	\iffalse
		\documentclass[12pt, a4paper, oneside, reqno]{article}
		\usepackage[a4paper, total={7.75in, 10.75in}]{geometry}
	\else
		\documentclass[10pt, a4paper, oneside, reqno]{article}
		\usepackage[a4paper, total={7in, 10.5in}]{geometry}
	\fi
\else 
	\documentclass[letterpaper, 10 pt, journal, twoside]{IEEEtran}
\fi 
\usepackage[noadjust]{cite}

\usepackage{amsmath, amssymb,mathrsfs, dsfont} 

\usepackage{bm,bbm}
\DeclareMathAlphabet{\mathbbold}{U}{bbold}{m}{n}

\usepackage{xspace}
\newcommand{\MATLAB}{\textsc{Matlab}\xspace}
\usepackage{tikzsymbols}
\usepackage{color}
\usepackage{graphicx,graphics} 
\usepackage{epstopdf} 
\usepackage{caption}
\usepackage{algorithm, algorithmicx, multirow, titlecaps} %
\usepackage{threeparttable}
\usepackage[english]{babel}
\usepackage[noend]{algpseudocode} 
\usepackage{tabularx}
\usepackage{booktabs}
\usepackage{multirow}
\usepackage{courier}
\usepackage{float}
\usepackage{makecell}
\setcellgapes{5pt}
\newcommand{\ra}[1]{\renewcommand{\arraystretch}{#1}}
\usepackage{array}
\usepackage{adjustbox}

\usepackage{authblk}
\usepackage{url} 

\newcommand{\mtiny}[1]{{\scalebox{.75}{#1}}}
\newcommand{\smtiny}[1]{{\scalebox{.63}{#1}}}
\newcommand{\stiny}[1]{{\scalebox{.5}{#1}}}

\newcommand{\supscrpsm}[1]{{\smtiny{$\mathrm{(}#1\mathrm{)}$}}}

\newcommand{\argmin}{{\mathrm{argmin}}}

\newcommand*{\minOp}{\operatornamewithlimits{min}\limits}

\newcommand*{\sumOp}{\operatornamewithlimits{\sum}\limits}
\newcommand*{\limOp}{\operatornamewithlimits{lim}\limits}

\newcommand{\tr}{{\smtiny{$\mathsf{T}$ }}\!}

\newcommand{\one}{\mathbf{1}}
\iftrue

\else

\fi

\newcommand{\vc}[1]{{ \mathrm{#1} }}
\newcommand{\mx}[1]{{ \mathrm{#1} }}

\newcommand{\drm}{\mathrm{d}}
\newcommand{\linspan}{\mathrm{span}} 

\newcommand{\inner}[2]{{ \langle {#1,#2} \rangle}}
\newcommand{\nth}{{\text{\tiny{th}}}}

\newcommand{\ellone}{\ell^{1}}
\newcommand{\ellinfty}{\ell^{\infty}}

\newcommand{\Lone}{L^{1}}
\newcommand{\Linfty}{L^{\infty}}
\newcommand{\Lp}{L^{p}}
\newcommand{\Lscrone}{\Lscr^{1}}
\newcommand{\Lscrinfty}{\Lscr^{\infty}}
\newcommand{\Lscrp}{\Lscr^{p}}

\newcommand{\Dscr}{{\mathscr{D}}}

\newcommand{\Gscr}{{\mathscr{G}}}
\newcommand{\Hscr}{{\mathscr{H}}}

\newcommand{\Lscr}{{\mathscr{L}}}

\newcommand{\Tscr}{{\mathscr{T}}}

\newcommand{\Vscr}{{\mathscr{V}}}
\newcommand{\Wscr}{{\mathscr{W}}}
\newcommand{\Xscr}{{\mathscr{X}}}

\newcommand{\Ccal}{{\mathcal{C}}}

\newcommand{\Ecal}{{\mathcal{E}}}

\newcommand{\Hcal}{{\mathcal{H}}}
\newcommand{\Ical}{{\mathcal{I}}}
\newcommand{\Jcal}{{\mathcal{J}}}

\newcommand{\Rcal}{{\mathcal{R}}}
\newcommand{\Scal}{{\mathcal{S}}}

\newcommand{\Ycal}{{\mathcal{Y}}}

\newcommand{\Nbb}{{\mathbb{N}}}

\newcommand{\Rbb}{{\mathbb{R}}}

\newcommand{\Tbb}{{\mathbb{T}}}

\newcommand{\Xbb}{{\mathbb{X}}}

\newcommand{\Zbb}{{\mathbb{Z}}}

\usepackage{amsthm}

\newtheorem{theorem}{Theorem}
\newtheorem*{theorem*}{Theorem}
\newtheorem{definition}{Definition}
\newtheorem*{definition*}{Definition}
\newtheorem{assumption}{Assumption}

\newtheorem{corollary}[theorem]{Corollary}
\newtheorem{lemma}[theorem]{Lemma}
\newtheorem{remark}{Remark}
\newtheorem{example}{Example}
\newtheorem*{example*}{Example}

\newtheorem*{claim*}{Claim}
\newtheorem{problem}{Problem}
\newtheorem*{problem*}{Problem}

\newcommand\xqed[1]{%
	\leavevmode\unskip\penalty9999 \hbox{}\nobreak\hfill
	\quad\hbox{#1}}
\iftrue 
\usepackage[colorlinks=true]{hyperref}
	\usepackage{xcolor}
	\hypersetup{
		colorlinks,
		linkcolor={red!85!black},
		citecolor={blue!55!black},
		urlcolor={blue!80!black}
	}
	\allowdisplaybreaks
	
\fi
\newcommand{\dc}{\ell_{\text{\mtiny{$0$}}}}
\newcommand{\udelta}{\underline{\delta}}
\newcommand{\odelta}{\overline{\delta}}

\newcommand{\rth}{{\smtiny{($r$)}}}

\newcommand{\rthn}{{\smtiny{($r_n$)}}}
\newcommand{\rthm}{{\smtiny{($r_m$)}}}
\newcommand{\nD}{n_{\stiny{$\!\Dscr$}}}
\newcommand{\nS}{n_{\text{s}}}
\newcommand{\nI}{n_{\stiny{$\Ical$}}}

\newcommand{\Lu}[1]{\mx{L}^{\!\vc{u}}_{#1}}
\newcommand{\gS}{\vc{g}^{\stiny{$(\Scal)$}}}
\newcommand{\gtS}{{g}^{\stiny{$(\Scal)$}}}
\newcommand{\GS}{G^{\stiny{$(\Scal)$}} }
\newcommand{\gstar}{\vc{g}^{\star}}

\newcommand{\xstar}{\vc{x}^{\star}}
\newcommand{\xstart}[1]{x_{#1}^{\star}}
\newcommand{\sS}{\vc{s}^{\stiny{$(\Scal)$}}}

\newcommand{\vcf}{\vc{f}}
\newcommand{\vcg}{\vc{g}}
\newcommand{\vch}{\vc{h}}
\newcommand{\vcs}{\vc{s}}
\newcommand{\vcu}{\vc{u}}
\newcommand{\vcv}{\vc{v}}
\newcommand{\vcx}{\vc{x}}
\newcommand{\vcy}{\vc{y}}

\newcommand{\mxA}{\mx{A}}

\newcommand{\vca}{\vc{a}}

\newcommand{\Hilbert}{\Hscr}
\newcommand{\kernel}{\mathds{k}}

\newcommand{\Hk}{\Hscr_{\kernel}}

\newcommand{\TC}{\text{\mtiny{$\mathrm{TC}$}}}
\newcommand{\DC}{\text{\mtiny{$\mathrm{DC}$}}}
\renewcommand{\SS}{\text{\mtiny{$\mathrm{SS}$}}}
\newcommand{\kernelTC}{\kernel_{\TC}}
\newcommand{\kernelDC}{\kernel_{\DC}}
\newcommand{\kernelSS}{\kernel_{\text{\mtiny{$\mathrm{SS}$}}}}
\newcommand{\psiTC}{\psi_{\TC}}
\newcommand{\psiDC}{\psi_{\DC}}
\newcommand{\psiSS}{\psi_{\text{\mtiny{$\mathrm{SS}$}}}}
\newcommand{\nuTC}{\nu_{\TC}}
\newcommand{\nuDC}{\nu_{\DC}}
\newcommand{\nuSS}{\nu_{\text{\mtiny{$\mathrm{SS}$}}}}
\newcommand{\barnu}{\bar{\nu}}
\newcommand{\barnuTC}{\bar{\nu}_{\TC}}
\newcommand{\barnuDC}{\bar{\nu}_{\DC}}
\newcommand{\barnuSS}{\bar{\nu}_{\text{\mtiny{$\mathrm{SS}$}}}}

\newcommand{\phiu}[1]{\varphi_{#1}^{\text{\rm{(u)}}}}

\usepackage{multirow}
\usepackage{booktabs}
\usepackage{array}
\usepackage{lineno}
\title{\LARGE Kernel-based Impulse Response Identification with Side-Information on Steady-State Gain}

\title{\LARGE Kernel-based Impulse Response Identification with Side-Information on Steady-State Gain}
\author{M.~Khosravi and R.~S.~Smith
	\thanks{Mohammad Khosravi is with  Delft Center for Systems and Control, Delft University of Technology,
		Delft, The Netherlands (email: Mohammad.khosravi@tudelft.nl).}
	\thanks{Roy~S.~Smith is with with Automatic Control Laboratory, ETH Z\"urich, Switzerland (email: rsmith@control.ee.ethz.ch).}
}

\begin{document}
\maketitle
\begin{abstract} %
In this paper, we consider the problem of system identification when side-information is available on the steady-state (or DC) gain of the system. We formulate a general nonparametric identification method as an infinite-dimensional constrained convex program over the reproducing kernel Hilbert space (RKHS) of stable impulse responses. The objective function of this optimization problem is the empirical loss regularized with the norm of RKHS, and the constraint is considered for enforcing the integration of the steady-state gain side-information. The proposed formulation addresses both the discrete-time and continuous-time cases. We show that this program has a unique solution obtained by solving an equivalent finite-dimensional convex optimization. This solution has a closed-form when the empirical loss and regularization functions are quadratic and exact side-information is considered. We perform extensive numerical comparisons to verify the efficiency of the proposed identification methodology.	
\end{abstract}

\begin{IEEEkeywords}
Kernel-based identification, side-information, steady-state gain.
\end{IEEEkeywords}
\section{Introduction}\label{sec:introduction}
System identification is a well-established research area on the theory and techniques of creating appropriate mathematical abstractions for the dynamical systems using their measurement data \cite{zadeh1956identification}. According to the importance and numerous applications of system identification in different fields of science and technology, it has received a significant deal of attention \cite{LjungBooK2}. In various situations, identifying a dynamical system can be beyond a mere model fitting to the data, and additionally, we may  need to include particular known features and attributes of the system into the model. More precisely, together with the measurement data, we might be provided with certain so-called side-information, which is indeed a specific qualitative or quantitative knowledge to be incorporated in the identified model of the system. This side-information can originate from various sources, e.g., a general understanding of the intrinsic physical nature of the system, or from the observed behaviors in experimental or historical data \cite{johansen1998constrained}.
Integrating side-information  can improve the identification performance by rejecting spurious model candidates, which are common when the measurement data is scarce, highly noise-contaminated, or generated by insufficient excitation \cite{trnka2009subspace}.

Various forms of side-information such as stability, dissipativity, region of attraction, and many others are considered in identifying
nonlinear dynamical systems \cite{umenberger2018specialized,khosravi2021ROA,ahmadi2020learning,khosravi2021grad,khosravi2022Koopman}. On the other hand, the key role of linear systems in practice has led to increased research on how to integrate different sorts of side-information in their identification \cite{lyzell2009handling}. For example, the incorporation of structural side-information, that is the available knowledge on the configuration and types of the subsystems, has been studied by imposing specific model structural constraints such as being circulant \cite{massioni2008subspace}.
Due to the importance of frequency domain analysis in controller design, a variety of  relevant properties are included in the identification procedure, e.g., the location of poles of the system  \cite{miller2013subspace}, 
phase constraints \cite{mckelvey2005estimation}, the peak in frequency response \cite{hoagg2004subspace}, moments and derivatives of the transfer function \cite{inoue2019subspace}, 
positive-realness \cite{goethals2003identification}, and the passivity of the system 
\cite{rodrigues2021novel}. 
Identification with side-information such as positivity or being compartmental are considered in \cite{de2002identification,benvenuti2002model}. 
The side-information on the low internal complexity of the system is included by means of sparsity promoting regularizations  such as the rank and the nuclear norm of associated Hankel matrices \cite{smith2014frequency,fazel2013hankel}, or by employing  atomic norm regularization applied on specific atomic linear representations of the system \cite{shah2012linear}. 
The stability side-information is integrated into the subspace identification method by ensuring that the poles of the identified models are inside the unit disc \cite{lacy2003subspace,van2001identification}. 
The kernel-based system identification approach \cite{pillonetto2010new}, besides addressing model order selection, robustness, and bias-variance trade-off issues, opened new avenues for the integration of various types of side-information \cite{ljung2020shift,khosravi2021robust,pillonetto2014kernel,khosravi2022Lut}, including stability, dissipativity, resonant frequencies, smoothness of the impulse response, oscillatory behaviors, relative degree, exponential decay of the impulse response, structural properties,  and the presence of fast and slow poles \cite{fujimoto2017extension,chen2012estimation,darwish2018quest,chen2018kernel,marconato2016filter,risuleo2017nonparametric,risuleo2019bayesian,everitt2018empirical,khosravi2021FDI}. Moreover, the incorporation of positivity and internal low-complexity are revisited in this framework \cite{khosravi2020low,khosravi2020regularized,pillonetto2016AtomicNuclearKernel,khosravi2019positive, zheng2021bayesian,khosravi2021POS}.

The steady-state gain information has particular importance from the control perspective, e.g., in closed-loop design and model predictive control   \cite{grosdidier1985closed,campo1994achievable,zheng1993robust}.
This information may be obtained, in exact or approximate form, through the structure of the system, from the experimental or historical collected data, or by designing and performing suitable experiments.
Hence, integrating the steady-state gain side-information into the identified model is of particular interest.
To this end,  various heuristics are introduced   
based on the subspace identification approach  \cite{trnka2009subspace,alenany2011improved,yoshimura2019system,abe2016subspace,markovsky2017subspace}.
Indeed, to identify a finite impulse response (FIR) model for the system, 
the subspace method can be employed in the multi-step ahead prediction form \cite{trnka2009subspace,alenany2011improved}. 
Following this and using a Bayesian approach, the steady-state gain side-information can be encoded in the covariance of the prior distribution \cite{trnka2009subspace}.
On the other hand, a frequentist framework is employed in \cite{alenany2011improved, yoshimura2019system, abe2016subspace, markovsky2017subspace}, where the steady-state gain side-information is incorporated by imposing linear constraints.
Moreover, to leverage the previously mentioned advantages of the kernel-based approach, Bayesian FIR estimation methods are proposed in \cite{fujimoto2018kernel,tan2020kernel}, where kernel-based priors are employed and the steady-state gain side-information is integrated into the resulting estimation problem. The identification scheme in \cite{fujimoto2018kernel} first estimates the step response of the system, and then, the impulse response is obtained via a na\"ive discrete derivative calculation, which  is prone to numerical imprecision and instability. One the other hand, while the method introduced in \cite{tan2020kernel} improves the estimation performance approach in \cite{fujimoto2018kernel}, the proposed formulation 
is incapable of including deterministic information on the steady-state gain of the system.
The identification approaches discussed 
are only applicable when a large set of high-quality data is available. 
Furthermore, they are limited to relatively short FIR estimation and fast decaying dynamics. 
Therefore, these estimation methodologies are not suitable for infinite impulse responses (IIR) and continuous-time systems, particularly when the dynamics have a very slowly decaying impulse response and considerably long memory.

In this paper, we develop a nonparametric identification approach where the side-information on the steady-state gain of the system is integrated into the proposed scheme. To leverage the powerful framework of kernel-based identification, we employ RKHS of stable impulse responses as the hypothesis space \cite{pillonetto2014kernel, chen2018stability}, enabling the formulation of
the problem for the continuous-time and discrete-time cases together. The identification problem is expressed as a constrained optimization where a generic regularized empirical loss is minimized subject to a suitably designed constraint encoding the available side-information on the steady-state gain of the system. According to the frequentist approach  employed, the resulting formulation is flexible, e.g., one can address the issue of outliers by defining the empirical loss based on the Huber function and its variants. We show that the steady-state gain linear functional is continuous on the employed RKHS, which implies that the problem is well-defined by guaranteeing the existence and uniqueness of the solution. 
For the initial infinite-dimensional formulation of the identification problem, we derive an equivalent finite-dimensional convex program with a unique solution.
This solution has a closed-form when exact side-information is considered, and the empirical loss and regularization functions used are quadratic.  Furthermore, we provide results for improving the computational complexity of the presented approach by obtaining the closed form of quantities used in the algorithm. We perform extensive numerical simulations confirming the efficacy of the proposed identification method.

\section{Notations and Preliminaries}
The set of natural numbers, the set of non-negative integers, the set of real numbers, the set of non-negative real numbers,  the $n$-dimensional Euclidean space and the space of $n$ by $m$ real matrices are denoted by $\Nbb$, $\Zbb_+$,  $\Rbb$, $\Rbb_+$,  $\Rbb^n$ and $\Rbb^{n\times m}$, respectively.
The $i^\nth$ entry of vector $\vca$ is denoted by  $[\vca]_{(i)}$, and the entry of matrix $\mxA$ at the $i^\nth$ row and the  $j^\nth$ column is denoted by $[\mxA]_{(i,j)}$.
To handle discrete and continuous time in the same formulation, 
$\Tbb$ denotes either $\Zbb_+$ or $\Rbb_+$, and $\Tbb_{\pm}$ is the set of scalars $t$ where either $t\in\Tbb$ or $-t\in\Tbb$. 
Given measure space $\Xscr$, the space of measurable functions $g:\Xscr\to \Rbb$ is denoted by $\Rbb^{\Xscr}$.
The element $\vc{u}\in\Rbb^{\Xscr}$ is shown entry-wise as $\vcu=(u_x)_{x\in\Xscr}$, or equivalently as $\vcu=\big(u(x)\big)_{x\in\Xscr}$.
Depending on the context of discussion, $\Lscrinfty$ refers either to $\ellinfty(\Zbb)$ or $\Linfty(\Rbb)$,  i.e., the space of bounded signals. 
Similarly, $\Lscrone$ is either $\ellone(\Zbb_+)$ or $\Lone(\Rbb_+)$,  i.e., the space of stable impulse responses. 
For $p\in\{1,\infty\}$, the norm in $\Lscrp$ is denoted by $\|\cdot\|_{p}$.
Given $\Vscr\subseteq \Xbb$, the linear span of $\Vscr$, denoted by $\linspan\Vscr$, is a linear subspace of $\Xbb$ containing linear combination of the elements of $\Vscr$.
Let $\Ycal$ be a set and $\Ccal\subseteq\Ycal$. 
We define the function $\delta_{\Ccal}$ as
$\delta_{\Ccal}(y) = 0$, if $y\in\Ccal$, and $\delta_{\Ccal}(y) = \infty$, otherwise.
Similarly, function $\mathbf{1}_{\Ccal}$ is defined as $\mathbf{1}_{\Ccal}(y) = 1$, if $y\in\Ccal$ and $\mathbf{1}_{\Ccal}(y) = 0$, otherwise.
With respect to each bounded signal $\vcu=(u_s)_{s\in \Tbb_{\pm}}\in \Lscrinfty$ and each $t\in\Tbb_{\pm}$, the linear map $\Lu{t}:\Lscrone\to\Rbb$ is defined as  
$\Lu{t}(\vcg) := \sum_{s\in \Zbb_+}g_s u_{t-s}$, when $\Tbb=\Zbb_+$, and 
$\Lu{t}(\vcg) := \int_{\Rbb_+}\!\!g_s u_{t-s}\ \!\drm s$, when $\Tbb=\Rbb_+$.

\section{Identification with Steady-State Gain Side-Information} \label{sec:problem_statement}
Let $\Scal$ be a stable LTI system 
with 
impulse response  $\gS:=(\gtS_t)_{t\in\Tbb}\in\Rbb^\Tbb$, where $\Tbb :=\Zbb_+$, for the case of discrete-time, and, $\Tbb :=\Rbb_+$, for the case of continuous-time.
The \emph{steady-state gain} 
system $\Scal$ is equal to $\dc(\gS)$, where $\ell_0$ is a real-valued linear operator defined on the space of stable impulse responses as following
\begin{equation}\label{eqn:dc_gain}
	\dc(\vcg) := 
	\begin{cases}
		\sum_{t\in \Zbb_+}g_t, 
		& \text{ if } \Tbb=\Zbb_+,\\	
		\int_{\Rbb_+}\!\!g_t \ \!\drm t, 
		& \text{ if } \Tbb=\Rbb_+,\\
	\end{cases}
\end{equation}
for any $\vcg=(g_t)_{t\in\Tbb}\in\Lscrone$.

Let $\vc{u}=(u_t)_{t\in \Tbb}$ be a bounded signal 
applied to the input of system  $\Scal$, and the corresponding output be measured at time instants 
\begin{equation}
\Tscr:=\{t_i\ \!|\! \ i=1,\ldots,\nD\},
\end{equation}
where $\nD\in\Nbb$ denotes the number of measurement samples.
From the definition of $\Lu{t}$, the measured output of the system at time instant $t\in\Tscr$, denoted by $y_t$, is 
\begin{equation}\label{eqn:output_sys_S}
	y_t := \Lu{t}(\gS)+w_t, \qquad t\in\Tscr,
\end{equation}
where  $\{w_t|t\in\Tscr\}$ are the measurement uncertainty.
Consequently, we are provided with the set of input-output data $\Dscr$ defined as 
\begin{equation}\label{eqn:DataSet}
	\Dscr = \{(u_t,y_t)\ | \ t\in\Tscr \}.
\end{equation} 
In addition to $\Dscr$, suppose that we know the steady-state gain of the system. Accordingly, one may ask whether the given steady-state gain side-information is naturally preserved and encoded in the identification of system $\Scal$. We elucidate this issue in the following demonstrative numerical example.
\begin{example*}\label{exm:pf}\normalfont
	Consider continuous-time system $\Scal$ described by the following transfer function
	\begin{equation} \label{eqn:sys_example}
		\GS(s) = \frac{s+2}{s^2+s+2},
	\end{equation}
	with the step response denoted by $\sS$.
	The system is initially at rest, and then actuate it by a random switching pulse signal in the time interval $[0,100]$. 
	To obtain the set of data $\Dscr$ in \eqref{eqn:DataSet}, the output of system is uniformly measured with the sampling frequency of $2$\,Hz and the signal-to-noise ratio (SNR) of $20$\,dB.
	Furthermore, let the steady-state gain of the system be given, i.e., we know that $\GS(0)=1$.
	The impulse response of the system can be estimated using \emph{direct} and \emph{indirect} methods \cite{garnier2008direct}.
	In the direct approach, we use the \texttt{tfsrivc} function provided by \textsc{CONTSID Toolbox} \cite{garnier2018contsid} with the known order of the system. 
	Let $\hat{\vcg}_1$ and $\hat{\vcs}_1$ respectively denote the impulse response and the step response of the resulting estimated system.
	Also, we identify the system indirectly by employing the \texttt{n4sid} function available in \MATLAB's \textsc{System Identification Toolbox} \cite{ljung2012version}  to estimate a discrete-time model, and subsequently, the continuous-time impulse response is obtained from a linear interpolation of the discrete-time estimate. 
	Let the resulting impulse and step responses be denoted by $\hat{\vcg}_2$ and $\hat{\vcs}_2$, respectively.
	As shown in Figure \ref{fig:example}, the steady-state values for $\hat{\vcs}_1$ and $\hat{\vcs}_2$ are respectively $0.85$ and $1.22$, meaning that the steady-state gains have a $15\%$ and a $22\%$ error.	
	Consequently, one can observe that the estimated models do not take into account the steady-state gain side-information.
\end{example*}
\begin{figure}[t]
	\centering
	\includegraphics[width =0.45\textwidth]{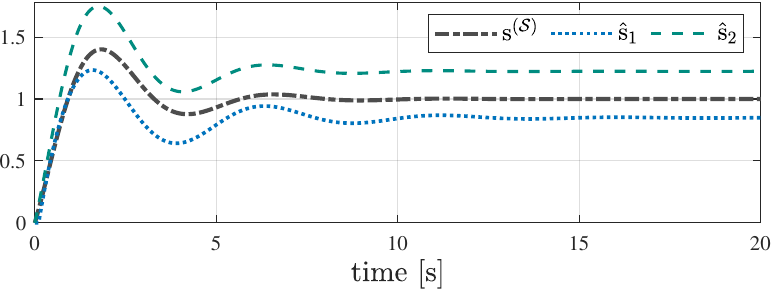}
	\caption{The step responses for system $\Scal$ and the estimated models.}
	\label{fig:example}
\end{figure}
Motivated by this example, the main problem discussed in this paper is the identification with side-information on the steady-state gain of the system. More precisely, we address the identification problem introduced below.
\begin{problem}
	\label{prob:ID_w_DC-gain_SI}
	Given the set of data $\Dscr$, estimate the impulse response of stable system $\Scal$ satisfying the side-information $\dc(\gS)\in [\udelta,\odelta]$, where $\udelta$ and $\odelta$ are given bounds for the steady-state gain of system $\Scal$. 
\end{problem}
Compared to the above example, designed to elaborate on the rationale of our discussion and its importance, Problem \ref{prob:ID_w_DC-gain_SI} addresses the more common scenarios in practice where the available side-information on the steady-state gain is imprecise and provided in the form of interval $[\udelta,\odelta]$. When the steady-state gain of the system is known to be exactly equal to $\delta$ that might be any arbitrary value in $\Rbb$, we set $\udelta=\odelta=\delta$.
The precise formulation of Problem \ref{prob:ID_w_DC-gain_SI} is discussed in the next section.

\section{The Estimation Problem: Existence and Uniqueness of the Solution}
\label{sec:MainOpt}
In this section, we formulate a constrained regularized empirical loss minimization to address estimation Problem~\ref{prob:ID_w_DC-gain_SI}. For this purpose, in addition to an appropriate objective function, we introduce a suitable hypothesis space characterizing the feasible set of the optimization problem and a constraint encoding the side-information about the steady-state gain of the system. Furthermore, we study the existence and uniqueness property for the solution of the resulting problem.
\subsection{Stable Reproducing Kernel Hilbert Spaces}
The hypothesis space taken for the estimation problem is  a  {\em reproducing kernel Hilbert space} (RKHS) \cite{berlinet2011reproducing}, which contains stable impulse responses.
Based on the structure of RKHS, we can investigate the problem and obtain a tractable approach for solving the estimation problem. 
These features are provided by the \emph{kernel function}, which characterizes the RKHS uniquely and completely.
\begin{definition}[\cite{berlinet2011reproducing}]
\label{def:kernel_and_section}
Let $\kernel:\Tbb\times\Tbb\to \Rbb$ be a non-zero symmetric measurable function. Then, we say $\kernel$ is a \emph{positive-definite kernel}, or simply \emph{kernel}, if 
we have 
\begin{equation}
\sum_{i=1}^{m}\sum_{j=1}^{m}
a_i\kernel(t_i,t_j)a_j\ge 0,
\end{equation} 
for any $m\in\Nbb$, $t_1,\ldots,t_m\in\Tbb$ and $a_1,\ldots,a_m\in\Rbb$.
If $\kernel$ is assumed to be continuous when $\Tbb=\Rbb_+$, then it is called a {\em Mercer kernel}.
The {\em section} of kernel $\kernel$ at $t\in\Tbb$, denoted by $\kernel_{t}$,
is the function defined as $\kernel(t,\cdot):\Tbb\to\Rbb$.
\end{definition}
\begin{theorem}[\cite{berlinet2011reproducing}]\label{thm:kernel_to_RKHS_def}
Given positive-definite kernel $\kernel:\Tbb\times\Tbb\to \Rbb$, there exists a unique Hilbert space $\Hk\subseteq \Rbb^{\Tbb}$ equipped with inner product $\inner{\cdot}{\cdot}_{\Hk}$, called a \emph{RKHS with kernel} $\kernel$, such that, for any $t\in\Tbb$, we have
\begin{itemize}
\item[i)] $ \kernel_t\in\Hk$, and
\item[ii)] $\inner{\vc{g}}{ \kernel_{t}}_{\Hk}=g_t$, for all $\vc{g}=(g_t)_{t\in\Tbb}\in\Hk$.
\end{itemize} 
The second feature is called the {\em reproducing property}.
\end{theorem}
According to Theorem \ref{thm:kernel_to_RKHS_def}, a RKHS is completely characterized by the corresponding kernel. 
As we are interested in the stable impulse responses in the bounded-input-bounded-output (BIBO) sense, we need to employ a kernel such that we have $\Hk\subseteq\Lscrone$.
The following theorem provides a necessary and sufficient condition for this feature. 

\begin{theorem}[\cite{pillonetto2014kernel,carmeli2006vector}]
	Let $\kernel:\Tbb\times\Tbb\to \Rbb$ be a positive-definite kernel. Then, $\kernel$ is {\em stable}  if and only if, for any $\vc{u}=(u_s)_{s\in\Tbb}\in\Lscrinfty$,  
	we have
	\begin{equation}
	\sum_{t\in\Zbb_+}\bigg|\sum_{s\in\Zbb_+}u_s\kernel(t,s)\bigg|<\infty,
	\end{equation} 
	when $\Tbb=\Zbb_+$, and,
	\begin{equation}
	\int_{\Rbb_+}\bigg|\int_{\Rbb_+}u_s\kernel(t,s)\drm s\bigg|\drm t<\infty,
	\end{equation} 
	when $\Tbb=\Rbb_+$.
	The kernel $\kernel$ is said to be \emph{stable} when it satisfies this property.
\end{theorem}
The following stable kernels are frequently used in the literature \cite{pillonetto2014kernel}:
\begin{itemize}
	\item \emph{diagonally/correlated} (DC) kernel:
	\begin{equation}\label{eqn:DC_kernel}
		\kernelDC(s,t) =  \alpha^{\max(s,t)} \gamma^{|s-t|},
	\end{equation}
	\item \emph{tuned/correlated} (TC) kernel:
	\begin{equation}\label{eqn:TC_kernel}
		\kernelTC(s,t) = \alpha^{\max(s,t)}, 
	\end{equation}
	\item \emph{stable spline} (SS) kernel:	
	\begin{equation}\label{eqn:SS_kernel}
		\kernelSS(s,t) = \alpha^{\max(s,t)+s+t}-\frac{1}{3}\alpha^{3\max(s,t)},
	\end{equation}	
\end{itemize}
where $\alpha,\gamma\in\Rbb$ are such that $\alpha\in(0,1)$, $|\gamma|\in(0, \sqrt{\alpha^{-1}})$, if $\Tbb=\Zbb_+$, and, $\gamma\in(0,\sqrt{\alpha^{-1}})$, if $\Tbb=\Rbb_+$. 
Note that by setting $\rho=\sqrt{\alpha}\gamma$, the definition of DC kernel in \cite{pillonetto2014kernel} reduces to \eqref{eqn:DC_kernel}. Also, without loss of generality, we can drop the extra scaling factor introduced in \cite{pillonetto2014kernel}.

\subsection{Empirical Loss and Regularization Function}\label{ssec:empirical_loss_regularization}
Given the set of data $\Dscr$ and the hypothesis space $\Hk$, the \emph{empirical loss function}, 
$\Ecal: \Hk\to\Rbb_+$, can be defined as the sum of squared errors, i.e., for each $\vcg\in\Hk$, we have
\begin{equation}\label{eqn:Emprical_loss_sum_of_squared_error}
\Ecal(\vcg):=\sum_{i=1}^{\nD}\big(\Lu{t_i}(\vcg) - y_{t_i} \big)^2. 
\end{equation}
We can consider a more general form for the empirical loss function.
More precisely, let $\Ical:=\{i_k|k=1,\ldots,\nI\}$ be a subset of $\{1,\ldots,\nD\}$, $\vcy_{\Ical}$ be the vector defined as $\vcy_{\Ical}=[y_{t_i}]_{i\in\Ical}$, and 
$\ell:\Rbb^{\nI}\times \Rbb^{\nI}\to \Rbb_+$ be a given convex function. 
Accordingly, we define the generalized \emph{loss function}, $\Ecal_{\ell}:\Hk\to\Rbb_+$ as follows
\begin{equation}\label{eqn:Emprical_loss_ell}
\Ecal_{\ell}(\vcg):=\ell([\Lu{t_i}(\vcg)]_{i\in\Ical},\vcy_{\Ical}), \qquad \forall \vcg\in\Hk,
\end{equation}
where subscript $\ell$ is considered to highlight the role of function $\ell$ in the definition of $\Ecal_{\ell}$.
When $\Ical=\{1,\ldots,\nD\}$ and function $\ell:\Rbb^{\nD}\times\Rbb^{\nD}\to\Rbb_+$ is defined as $\ell(\vcv_1,\vcv_2)=\|\vcv_1-\vcv_2\|^2$, for any $\vcv_1,\vcv_2\in\Rbb^{\nD}$, the empirical loss $\Ecal_{\ell}$ reduces to the special case introduced in  \eqref{eqn:Emprical_loss_sum_of_squared_error}.
Also, to be robust with respect to outliers, one may take function  $\ell:\Rbb^{\nD}\times\Rbb^{\nD}\to\Rbb_+$ as
\begin{equation}
	\ell(\vcv_1,\vcv_2)=\sum_{i=1}^{\nD}L_\sigma(\left[\vcv_1\right]_{(i)}-\left[\vcv_2\right]_{(i)}),
	\quad
	\forall \vcv_1,\vcv_2\in\Rbb^{\nD},
\end{equation}
where, for given $\sigma\in\Rbb_+$, function $L_\sigma:\Rbb\to\Rbb_+$ is the \emph{Huber loss}  defined as follows
\begin{equation}\label{eqn:Huber}
	L_\sigma(e) = 
	\begin{cases}
		\frac{1}{2}e^2, &\text{ if } |e|\le \sigma,\\
		\sigma(|e|-\frac{1}{2}\sigma), &\text{ otherwise, }\\
	\end{cases}
\end{equation}
or, it can be the smoothed version of \eqref{eqn:Huber}, known as the \emph{pseudo-Huber} function \cite{maronna2019robust}, which is 
\begin{equation}\label{eqn:smoothed-Huber}
	L_\sigma(e) = (e^2+\sigma^2)^{\frac{1}{2}}-\sigma^2, \quad \forall e\in\Rbb.
\end{equation} 
The resulting empirical loss function $\Ecal_{\ell}$ is more suitable to the cases where output measurements are subject to noise disturbances with large outliers.

Since the hypothesis space is a RKHS endowed with kernel $\kernel$, we define a regularization term,  enforcing desired attributes such as stability, based on the norm in $\Hk$, i.e., $\|\cdot\|_{\Hk}$. More precisely, let $\rho:\Rbb_+\to\Rbb_+$ be a strictly increasing 
convex function. Then, the regularization function, $\Rcal:\Hk\to\Rbb_+$, is defined as $\Rcal(\vcg)=\rho(\|\vcg\|_{\Hk})$, .
Accordingly, the objective function for the estimation problem, $\Jcal: \Hk\to\Rbb_+$, is defined as following
\begin{equation}\label{eqn:Jcal}
\Jcal(\vcg):=\Ecal_{\ell}(\vcg)+\lambda \Rcal(\vcg), \qquad \forall \vcg\in \Hk,
\end{equation}
where $\lambda>0$ is the regularization  weight.
Note that, in addition to enforcing desired attributes such as stability, the regularization term helps avoiding over-fitting phenomena,  enhancing the numerical performance, and improving the bias-variance trade-off.
\subsection{Steady-State Gain Side-Information}
Define the set $\Gscr_{\kernel}([\udelta,\odelta])\subset\Hk$ as follows
\begin{equation}\label{eqn:Gcal}
\Gscr_{\kernel}([\udelta,\odelta]) := \Big\{\vcg\in\Hk\ \! \Big| \ \! \dc(\vcg)\in[\udelta,\odelta]\Big\}.
\end{equation}
The elements of $\Gscr_{\kernel}([\udelta,\odelta])$ are exactly the ones satisfying the side-information on the steady-state gain of the system. 
Therefore, the estimation Problem \ref{prob:ID_w_DC-gain_SI} is formulated as the following  optimization problem
\begin{equation}\label{eqn:opt_dc_gain_1}
\begin{array}{cl}
\minOp_{\vcg\in\Hk} & \ \Ecal_{\ell}(\vcg)+\lambda \Rcal(\vcg)\\
\mathrm{s.t.} & \ \vcg\in\Gscr_{\kernel}([\udelta,\odelta]).
\end{array}	
\end{equation}
The existence and uniqueness of the solution of optimization problem \eqref{eqn:opt_dc_gain_1} depends on the topological properties of set $\Gscr_{\kernel}([\udelta,\odelta])$ which is characterized by operator $\dc:\Hk\to\Rbb$.
To study these properties, we need the notion of integrable kernels \cite{pillonetto2014kernel}.

\begin{definition}\label{def:integrable_kernel}
The positive-definite kernel $\kernel:\Tbb\times\Tbb\to \Rbb$ is said to be \emph{integrable}
if
\begin{equation}
	\int_{\Rbb_+}\int_{\Rbb_+}|\kernel(s,t)|\ \! \drm s \drm t<\infty,
\end{equation}
when $\Tbb=\Rbb_+$, or, if 
\begin{equation}\label{eqn:abs_summable}
\sum_{s\in\Zbb_+}
\sum_{t\in\Zbb_+}
|\kernel(s,t)|<\infty,
\end{equation} 
when $\Tbb=\Zbb_+$.
\end{definition}
One can easily see that each of the kernels $\kernelTC$, $\kernelDC$ and $\kernelSS$ is integrable (see Appendix \ref{sec:appendix_proof_kTC_kDC_kSS_integrable}).
Moreover, for any integrable kernel $\kernel$ and any $\vc{u}=(u_s)_{s\in\Tbb}\in\Lscrinfty$,  
we have
\begin{equation*}
	\sum_{t\in\Zbb_+}\Big|\sum_{s\in\Zbb_+}u_s\kernel(t,s)\Big|
	\le
	\|\vcu\|_\infty\sum_{t\in\Zbb_+}\sum_{s\in\Zbb_+}|\kernel(t,s)|	
	<\infty,
\end{equation*} 
when $\Tbb=\Zbb_+$, and,
\begin{equation*}
	\int_{\Rbb_+}\!\bigg|\int_{\Rbb_+}\!u_s\kernel(t,s)\drm s\bigg|\drm t
	\le 
	\|\vcu\|_\infty
	\int_{\Rbb_+}\!\int_{\Rbb_+}\!|\kernel(t,s)|\drm s\drm t
	<\infty,
\end{equation*} 
when $\Tbb=\Rbb_+$. Accordingly, the integrable kernels are stable.
The main importance of integrable kernels in this paper is highlighted by the following theorems.
The next theorem is the cornerstone of this work.
\begin{theorem}\label{thm:dc_bounded}
	Let  $\kernel:\Tbb\times\Tbb\to \Rbb$ be an integrable Mercer kernel, and $\varphi_0=(\varphi_{0,t})_{t\in\Tbb}$ be defined as following
	\begin{equation}\label{eqn:phi_0t}
		\varphi_{0,t} =  
		\begin{cases}
		\sum_{s\in\Zbb_+}\kernel(t,s),   & \text{ if }\Tbb=\Zbb_+,\\		
		\int_{\Rbb_+}\kernel(t,s)\drm s, & \text{ if }\Tbb=\Rbb_+,\\		
		\end{cases}
	\end{equation}
	for any $t\in\Tbb$.
	Then, $\varphi_0$ is well-defined and $\varphi_0\in\Hk$. Moreover, we have 
	\begin{equation}\label{eqn:inner_phi_0_g}	
		\dc(\vcg)=\inner{\varphi_0}{\vcg}_{\Hk},\qquad \forall \vcg \in \Hk.
	\end{equation}
	Furthermore, one can see
	\begin{equation}\label{eqn:norm_phi_0}	
		\|\varphi_0\|_{\Hk}^2  
		=	
		\begin{cases}
			\sum_{s,t\in\Zbb_+}\kernel(t,s),   & \text{ if }\Tbb=\Zbb_+,\\		
			\int_{\Rbb_+\times\Rbb_+}	\kernel(t,s)\drm s\drm t, & \text{ if }\Tbb=\Rbb_+.\\		
		\end{cases}
	\end{equation}
\end{theorem}
\begin{proof}
See Appendix \ref{sec:appendix_proof_dc_bounded}.
\end{proof}
Theorem \ref{thm:dc_bounded} says that integrability of kernel $\kernel$ implies that the steady-state operator $\ell_0:\Hk\to\Rbb$ is a linear continuous functional.
Accordingly, throughout this paper, we assume $\kernel$ is an integrable kernel.
More precisely, we make the following assumption.
\begin{assumption}\label{ass:kernel_nonzero_integral}
The kernel	$\kernel$ is an integrable Mercer kernel, for which there exists $\tau\in\Tbb$ such that 
$\sum_{s\in\Zbb_+}\kernel(\tau,s)\ne 0$, when $\Tbb=\Zbb_+$, or,   
$\int_{\Rbb_+}\kernel(\tau,s)\drm s\ne 0$, when $\Tbb=\Rbb_+$.
\end{assumption}
The following theorem describes topological properties of set $\Gscr_{\kernel}([\udelta,\odelta])$, and together with Theorem~\ref{thm:dc_bounded}, provides the necessary foundation to guarantee the introduced problem is well-defined.
\begin{theorem}\label{thm:G_nonempty_closed_convex}
Let Assumption \ref{ass:kernel_nonzero_integral} hold. 
Then,   for any $\udelta$ and $\odelta$ such that $-\infty\le\udelta\le\odelta\le\infty$, the set $\Gscr_{\kernel}([\udelta,\odelta])$ is a non-empty, closed and convex subset of $\Hk$.
\end{theorem}
\begin{proof}
See Appendix~\ref{sec:appendix_proof_G_nonempty_closed_convex}.
\end{proof}
\subsection{From Infinite to Finite Dimension}
The optimization problem \eqref{eqn:opt_dc_gain_1} is defined over the infinite-dimensional Hilbert space $\Hk$.
In the following, we show that \eqref{eqn:opt_dc_gain_1} admits a unique solution in $\Hk$.
Furthermore, we introduce an equivalent convex
finite-dimensional program which provides a tractable approach to address \eqref{eqn:opt_dc_gain_1}.
To this end, we need to introduce additional definitions and mathematical properties for \eqref{eqn:opt_dc_gain_1}. 
Theorem \ref{thm:G_nonempty_closed_convex} has provided suitable properties for the feasible set of \eqref{eqn:opt_dc_gain_1}. 
The next assumption and lemma provide foundations to show that the objective function in 
\eqref{eqn:opt_dc_gain_1} has desired features which are latter employed in the main theorem of this paper to show the existence and uniqueness for the solution of \eqref{eqn:opt_dc_gain_1}.
\begin{assumption}\label{ass:Lu_bounded}
The operator $\Lu{\tau}:\Hk\to\Rbb$ is continuous for each $\tau\in\Tscr$. 	
\end{assumption}
When $\Tbb=\Rbb_+$ and $\vcu$ is a step function as in \eqref{eqn:u_pw}, one can show the continuity of $\Lu{\tau}$, for any $\tau\in\Rbb_+$, based on an  argument similar to the proof of Theorem \ref{thm:dc_bounded}.
Also, for the case of $\Tbb=\Zbb_+$, one can easily see that Assumption \ref{ass:Lu_bounded} holds if the system is initially at rest, or more generally, when $(u_t)_{t\le t_{\nD-1}}$ is finitely non-zero.
Given this assumption, we have the following theorem.
\begin{lemma}\label{thm:Lu_bounded}
	Let Assumption \ref{ass:Lu_bounded} hold. Then, for each $\tau\in\Tscr$,
	there exists $\phiu{\tau}=(\phiu{\tau,t})_{t\in\Tbb}\in\Hk$ such that
	\begin{equation}\label{eqn:Lu_bounded}
		\Lu{\tau}(\vcg) = \inner{\phiu{\tau}}{\vcg}_{\Hk}, \quad \forall\vcg\in\Hk.
	\end{equation}
	Furthermore, for any $t\in\Tbb$, 
	we have 
	\begin{equation}\label{eqn:phiu_tau_t}
		\phiu{\tau,t} =  
		\begin{cases}
			\int_{\Rbb_+}\kernel(t,s)u_{\tau-s}\drm s,
			&
			\text{ if } \Tbb=\Rbb_+,
			\\
			\sum_{s\in\Zbb_+}\kernel(t,s)u_{\tau-s},
			&
			\text{ if } \Tbb=\Zbb_+.
		\end{cases}
	\end{equation}
\end{lemma}
\begin{proof}
	See Appendix \ref{sec:appendix_proof_Lu_bounded}.
\end{proof}
Recall the index set $\Ical=\{i_k|k=1,\ldots,\nI\}$  introduced in Section~\ref{ssec:empirical_loss_regularization}.
For $k=1,\ldots,\nI$, let $\varphi_k$ be defined as $\phiu{t_i}$ with $i=i_k$.
Accordingly, we define matrices $\Phi$ and $\mxA$ respectively as 
\begin{equation}\label{eqn:Phi}
\Phi:=
\big[\inner{\varphi_i}{\varphi_j}_{\Hk}\big]_{i=0,j=0}^{\nI,\nI}
\in\Rbb^{(\nI+1)\times(\nI+1)},
\end{equation}
and 
\begin{equation}\label{eqn:A}
	\mxA:=
	\big[\inner{\varphi_i}{\varphi_j}_{\Hk}\big]_{i=1,j=0}^{\nI,\nI}
	\in\Rbb^{\nI\times(\nI+1)}.
\end{equation}
Note that $\mxA$ is a sub-matrix of $\Phi$ which contains the rows corresponding to the index set $\Ical$. 
Denote the columns of $\Phi$ by  $\vca_0,\ldots,\vca_{\nI}\in\Rbb^{\nI+1}$, i.e., we have $\Phi=[\vca_0,\ldots,\vca_{\nI}]$.
Since $\Phi$ is a symmetric matrix, one can see that $\mxA=[\vca_1,\ldots,\vca_{\nI}]^\tr\!$.

Given the above definitions and theorems, we can present our main theorem.
\begin{theorem}\label{thm:main_thm}
Let Assumption \ref{ass:kernel_nonzero_integral} and  
Assumption \ref{ass:Lu_bounded} hold. Then, the optimization problem \eqref{eqn:opt_dc_gain_1} admits a \emph{unique} solution $\gstar$.
Moreover, there exists
$\xstar=[\xstart{0},\ldots,\xstart{\nI}]^\tr\in\Rbb^{\nI+1}$ such that
\begin{equation}\label{eqn:gstar_parametric_form}
	\gstar = 
	\xstart{0}\varphi_0 + \ldots +\xstart{\nI}\varphi_{\nI}  = 
	\sum_{i=0}^{\nI}\xstart{i}\varphi_i.
\end{equation}
Furthermore, $\xstar$ is the solution of following convex program
\begin{equation}\label{eqn:opt_dc_gain_2}
	\begin{array}{cl}
		\minOp_{\vcx\in\Rbb^{\nI+1}} & \ \ell(\mxA\vcx,\vcy_{\Ical})
		+
		\lambda \Rcal\big((\vcx^\tr\Phi\vcx)^{\frac12}\big)\\
		\mathrm{s.t.} & \ \vca_0^\tr\vcx\in[\udelta,\odelta].
	\end{array}	
\end{equation}
\end{theorem}
\begin{proof}
See Appendix \ref{sec:appendix_proof_main_thm}.
\end{proof}

In the literature, it is common to employ 
the empirical loss \eqref{eqn:Emprical_loss_sum_of_squared_error} and the regularization function $\Rcal(\vcg)=\|\vcg\|_{\Hk}^2$.
When the steady-state gain of the system is known to be $\delta\in\Rbb$, the resulting impulse response estimation problem is as following
\begin{equation}\label{eqn:opt_dc_gain_3}\ra{1.2}
\begin{array}{cl}
\minOp_{\vcg\in\Hk} & \ \sumOp_{i=1}^{\nD}\big(\Lu{t_i}(\vcg) - y_{t_i} \big)^2+\lambda \big\|\vcg\big\|_{\Hk}^2
\\
\mathrm{s.t.} & \ \dc(\vcg) = \delta.
\end{array}	
\end{equation}
The next corollary provides a closed-form solution for this optimization problem.
\begin{corollary}\label{cor:main_thm_SE}
Under the assumptions of Theorem \ref{thm:main_thm}, the convex program \eqref{eqn:opt_dc_gain_3} has a \emph{unique} solution $\gstar$.
Moreover, there exist
$\xstar=[\xstart{0},\ldots,\xstart{\nD}]^\tr\in\Rbb^{\nD+1}$ and $\lambda\in\Rbb$ such that  $\gstar$ has the parametric form \eqref{eqn:gstar_parametric_form} and $[\xstar{}^\tr,\gamma^{\star}]^\tr$ is a solution of the following system of linear equations
\begin{equation}\label{eqn:QP_KKT}
	\begin{bmatrix}
	\mx{Q}&\vca_0\\	
	\vca_0^\tr&0	
	\end{bmatrix}
	\begin{bmatrix}
	\vcx\\\gamma
	\end{bmatrix}
	=
	\begin{bmatrix}
	\mxA^\tr\vcy\\
	\delta	
	\end{bmatrix},
\end{equation} 
where $\mx{Q}=\mxA^\tr\mxA+\lambda\Phi$ and $\vcy=[y_{t_i}]_{i=1}^{\nD}\in\Rbb^{\nD}$.
Furthermore, when $\varphi_0,\ldots,\varphi_{\nD}$ are linearly independent, we have
\begin{equation}\label{eqn:xstar}
\xstar = 
\mx{Q}^{-1}\mxA^\tr\vcy + 
\frac{\delta-\vca_0^\tr\mx{Q}^{-1}\mxA^\tr\vcy}
{\vca_0^\tr\mx{Q}^{-1}\vca_0}
\mx{Q}\vca_0.
\end{equation}
\end{corollary}	
\begin{proof}
See Appendix \ref{sec:appendix_proof_main_thm_SE}.	
\end{proof}	
\section{The Optimization Problem: Settings and Algorithm}
Based on Theorem \ref{thm:main_thm}, addressing the estimation problem \ref{prob:ID_w_DC-gain_SI}, or equivalently, optimization problem \eqref{eqn:opt_dc_gain_1}, reduces to solving the convex program \eqref{eqn:opt_dc_gain_2}. In this section, we discuss how to configure this optimization problem.

The main elements of optimization problems \eqref{eqn:opt_dc_gain_2} are $\vca_0$, $\mxA$, and $\Phi$.
We know that $\mxA$ is a sub-matrix of $\Phi$ and $\vca_0$ is the first column of $\Phi$. Hence, it suffices to calculate the matrix $\Phi$. According to \eqref{eqn:Phi}, the entries of $\Phi$ are inner products $\inner{\varphi_i}{\varphi_j}_{\Hk}$, for $i,j\in\{0,\ldots,\nD\}$. 
To obtain the value of these inner products, we need to calculate improper double integrals when $\Tbb=\Rbb_+$, or infinite double summations when $\Tbb=\Zbb_+$. 
In general, these calculations can be performed using techniques such as numerical integration.
On the other hand, these values can be obtained analytically in certain but fairly general situations. For example,   
when $\Tbb=\Zbb_+$ and the system is initially at rest, or,
when $\Tbb=\Rbb_+$, the standard kernels are employed and the input of the system is a step function.   
In the remainder of this section, the details of these calculations are discussed.
\subsection{Optimization Problem Configuration: Discrete-Time Case}
Let $\Tbb=\Zbb_+$ and the system be initially at rest, i.e., we have $u_t=0$, for $t<0$.
Also, let the measurement time instants  be  $\Tscr=\{0,1,\ldots,\nD-1\}$.
Given an integrable kernel $\kernel:\Zbb_+\times\Zbb_+\to\Rbb$ and an input $\vcu$, define matrices $\mx{K},\mx{T}_{\vc{u}}\in\Rbb^{\nD\times\nD}$  such that
\begin{equation}
[\mx{K}]_{(i,j)}=\kernel(i-1,j-1), \quad \forall i,j\in\{1,2,\ldots,\nD\},
\end{equation}
and
\begin{equation}
[\mx{T}_{\vc{u}}]_{(i,j)}=u_{i-j},  \quad \forall i,j\in\{1,2,\ldots,\nD\}.
\end{equation}
Following these definitions, we have the next theorem.
\begin{theorem}\label{thm:Phi_DT}
Let $\varphi\in\Rbb^{\nD}$ be the column vector defined as $\varphi:=[\varphi_{0,i}]_{i=0}^{\nD-1}$.
Then, we have
\begin{equation}
\Phi =
\begin{bmatrix}
\|\varphi_0\|^2
&
\varphi \mx{T}_{\vcu}^\tr\\
\mx{T}_{\vcu}\varphi
&
\mx{T}_{\vc{u}}\mx{K}\mx{T}_{\vc{u}}^\tr
\end{bmatrix}.
\end{equation}
\end{theorem}	
\begin{proof}
See Appendix \ref{sec:appendix_proof_Phi_DT}.
\end{proof}
\begin{remark}\label{rem:Phi_DT}
Appendix \ref{sec:phi_0t_norm_phi_0_kTC_kDC_kSS} provides $\varphi_0$ and $\|\varphi_0\|^2_{\Hk}$ for the standard kernels introduced in \eqref{eqn:DC_kernel}, \eqref{eqn:TC_kernel}, and \eqref{eqn:SS_kernel}, when $\Tbb=\Zbb_+$.
\end{remark}	
\subsection{Optimization Problem Configuration: Continuous-Time Case}
The set of step functions is dense in $\Lp(\Rbb)$, for $p\in[1,\infty)$, and also, any function in $\Linfty(\Rbb)$ is an almost everywhere the point-wise limit of a sequence of step functions \cite{brezis2011functional}. In other words, any signal of interest can be approximated arbitrarily closely by step functions.
Accordingly, for the case of $\Tbb=\Rbb_+$, one can assume that the input signal $\vc{u} = (u_t)_{t\in\Rbb_+}$ is a step function. More precisely, there exist $\nS\in\Nbb$ real scalars $\xi_1,\ldots,\xi_{\nS}$
and a finite increasing sequence $(s_0,s_1,\ldots,s_{\nS})$ in $\Rbb_+$ such that we have
\begin{equation}\label{eqn:u_pw}
	u_t=\sum_{i=0}^{\nS-1} \xi_{i+1} \one_{[s_i,s_{i+1})}(t), \quad \forall t\in\Rbb_+.
\end{equation}
For $\vcu= (u_t)_{t\in\Rbb_+}$ given in \eqref{eqn:u_pw},  a closed-form for $\phiu{\tau}$ can be introduced. 
To this end, we require the function $\psi:\Rbb_+\times\Rbb_+\times\Rbb_+\to\Rbb$ defined as 
\begin{equation}\label{eqn:psi_def_general}
	\psi(t,a,b):= 	\int_{a}^{b}\kernel(t,s)\drm s, 
\end{equation} 
for any $a,b,t\in\Rbb_+$.
This function is denoted by $\psiTC$, $\psiDC$ or $\psiSS$ respectively when 
$\kernel$ is $\kernelTC$, $\kernelDC$ or $\kernelSS$.
\begin{theorem}\label{thm:phiu_tau_t_general}
	For any $t\in\Rbb_+$, we have
	\begin{equation}
		\label{eqn:phiu_taut_pw}
		\phiu{\tau,t} =
		\sum_{i=0}^{\nS-1} \xi_{i+1} \psi(t,\bar{s}_{i+1}(\tau),\bar{s}_{i}(\tau)),
	\end{equation}
	where, for  $i=0,\ldots,\nS$, function $\bar{s}_{i}:\Rbb_+\to\Rbb_+$ is defined as $\bar{s}_{i}(\tau):=\max(\tau-s_i,0)$, for any $\tau\in\Rbb_+$.
\end{theorem}
\begin{proof}
See Appendix \ref{sec:appendix_proof_phiu_tau_t_general}.
\end{proof}	
For the standard kernels defined in \eqref{eqn:DC_kernel}, \eqref{eqn:TC_kernel} and \eqref{eqn:SS_kernel}, 
one can obtain the closed-form of $\phiu{\tau}$ using \eqref{eqn:phiu_taut_pw}.
To this end, we need $\psiTC$, $\psiDC$, and $\psiSS$ which are provided by the next theorem. 
\begin{theorem}\label{thm:psi_TC_DC_SS}
	Define the function $\eta:\Rbb_+\times\Rbb_+\times\Rbb_+\to\Rbb$ as 
	\begin{equation}
		\eta(s,\tau_1,\tau_2) = \min(\max(t,\tau_1),\tau_2),	
	\end{equation}
	for any $s,\tau_1,\tau_2\in\Rbb_+$, and let $t,a,b\in\Rbb_+$ such that $a\le b$.
	Then, we have
	\begin{align}\label{eqn:psi_TC}
		\!\!\!\!
		\psiTC(t,a,b) & \!=\!
		\big(\eta(t,a,b)-a\big)\alpha^t
		+	\frac{\alpha^b-\alpha^{\eta(t,a,b)}}
		{\ln(\alpha)},
		\\
		\!\!\!\!
		\psiDC(t,a,b) & \!=\!
			\frac{\gamma^{-a}-\gamma^{-\eta(t,a,b)}}{\ln(\gamma)}
			(\alpha\gamma)^t
			\notag
			\\&\qquad
			+	\frac{(\alpha\gamma)^b-(\alpha\gamma)^{\eta(t,a,b)}}{\ln(\alpha\gamma)}
			\gamma^{-t},
		\label{eqn:psi_DC}
		\\		
		\!\!\!\!
		\psiSS(t,a,b)& \!=\!
			\frac{\alpha^{\eta(t,a,b)}-\alpha^{a}}{\ln(\alpha)}
			\alpha^{2t}
			+
			\frac{\alpha^{2b}-\alpha^{2\eta(t,a,b)}}{2\ln(\alpha)}
			\alpha^{t}
			\notag\\& \ \ \
			-\frac13\big(\eta(t,a,b)-a\big)\alpha^{3t}
			\!-\!	\frac{\alpha^{3b}-\alpha^{3\eta(t,a,b)}}{9\ln(\alpha)}.
	\label{eqn:psi_SS}
	\end{align}
\end{theorem}
\begin{proof}
See Appendix \ref{sec:appendix_proof_psi_TC_DC_SS}.
\end{proof}
Similar to the previous theorem, one can obtain the closed-form of $\varphi_{0}$ for the standard kernels \eqref{eqn:DC_kernel}, \eqref{eqn:TC_kernel} and \eqref{eqn:SS_kernel}. 
\begin{theorem}\label{thm:phi_0t_CT}
For any $t\in \Rbb_+$, we have
\begin{align}
	\label{eqn:phi_0t_kTC_CT}
	\varphi_{\TC,0,t} &=  
	\left(t-\frac{1}{\ln(\alpha)}\right)\alpha^t,
	\\
	\label{eqn:phi_0t_kDC_CT}
	\varphi_{\DC,0,t} &=  
	-\Big[
	\frac{(1-\gamma^t)}{\ln(\gamma)}
	+
	\frac{1}{\ln(\alpha\gamma)}\Big]\alpha^{t},
	\\
	\label{eqn:phi_0t_kSS_CT}
	\varphi_{\SS,0,t} &= 
	\Big[\frac{11}{18 \ln(\alpha)}\alpha^{t}-\frac{1}{\ln(\alpha)}-\frac{t\alpha^{t}}{3}\Big]\alpha^{2t}.
\end{align}
\end{theorem}
\begin{proof}
See Appendix \ref{sec:appendix_proof_phi_0t_CT}.	
\end{proof}	
Given $\{\varphi_i\}_{i=0}^{\nD}$, we can obtain $\inner{\varphi_i}{\varphi_j}_{\Hk}$, for $i,j=0,\ldots,\nD$.
To this end, define functions 
$\nu:\Rbb_+\times\Rbb_+\to\Rbb$ 
and 
$\barnu:\Rbb_+\to\Rbb$ 
respectively as
\begin{equation}\label{eqn:nu_def_general}
	\nu(x,y) := 
	\int_{0}^{x}\!
	\int_{0}^{y}\! 
	\kernel(s,t)
	\drm t \drm s,
	\qquad \forall x,y\in\Rbb_+,
\end{equation}
and
\begin{equation}\label{eqn:barnu_def_general}
	\barnu(x) := 
	\int_{0}^{x}\!
	\int_{0}^{\infty}\! 
	\kernel(s,t)
	\drm t  \drm s,
	\qquad \forall x\in\Rbb_+.
\end{equation}
When $\kernel$ is one of the standard kernels \eqref{eqn:DC_kernel}, \eqref{eqn:TC_kernel}  and \eqref{eqn:SS_kernel}, 
we include a suitable subscript in $\nu$ and $\barnu$ to indicate the corresponding kernel.
For each $i,j \in \{0,1,\ldots,\nS-1\}$, let functions $\kappa_{ij}:\Rbb_+\times\Rbb_+\to\Rbb_+$
and
$\bar{\kappa}_{i}:\Rbb_+\times\Rbb_+\to\Rbb_+$ be defined
such that, for any $\tau,\tau_1,\tau_2\in\Rbb_+$, we have  
\begin{equation}\label{eqn:kappa}
	\begin{split}
		&\kappa_{ij}(\tau_1,\tau_2) = 
		\nu\big(\bar{s}_{i}(\tau_1),\bar{s}_{j}(\tau_2)\big) 
		-
		\nu\big(\bar{s}_{i+1}(\tau_1),\bar{s}_{j}(\tau_2)\big)
		\\&\quad-
		\nu\big(\bar{s}_{i}(\tau_1),\bar{s}_{j+1}(\tau_2)\big) 
		+
		\nu\big(\bar{s}_{i+1}(\tau_1),\bar{s}_{j+1}(\tau_2)\big), 
	\end{split}
\end{equation}
and
\begin{equation}\label{eqn:barkappa}
	\bar{\kappa}_{i}(\tau) = 
	\barnu\big(\bar{s}_{i}(\tau)\big) 
	-
	\barnu\big(\bar{s}_{i+1}(\tau)\big).
\end{equation}
Based on these definitions, the next theorem presents the closed-form for 
$\inner{\varphi_{0}}{\phiu{\tau}}_{\Hk}$
and
$\inner{\phiu{\tau_1}}{\phiu{\tau_2}}_{\Hk}$.
\begin{theorem}\label{thm:inner_phi_0i_ij_CT}
For any $\tau,\tau_1,\tau_2\in\Rbb_+$, we have
\begin{equation*}
	\begin{array}{rcl}
	\inner{\varphi_{0}}{\phiu{\tau}}_{\Hk}
	&=& 
	\sum_{i=0}^{\nS-1}
	\xi_{i+1} \bar{\kappa}_{i}(\tau).
	\\
	\inner{\phiu{\tau_1}}{\phiu{\tau_2}}_{\Hk}
	&=& 
	\sum_{i=0}^{\nS-1}\sum_{j=0}^{\nS-1}
	\xi_{i+1}\xi_{j+1} \kappa_{ij}(\tau_1,\tau_2).
	\end{array}
\end{equation*}
\end{theorem}
\begin{proof}
See Appendix \ref{sec:appendix_proof_inner_phi_0i_ij_CT}.	
\end{proof}	
Due to \eqref{eqn:kappa} and \eqref{eqn:barkappa}, in order to employ Theorem~\ref{thm:inner_phi_0i_ij_CT} to calculate the entries of $\Phi$,  we need to obtain the functions $\nu$ and $\barnu$, which can be done in general using numerical techniques.
However, the closed-form of $\nu$ and $\barnu$ can be explicitly  derived for the standard kernels. 
\begin{theorem}\label{thm:nu_TC_DC_SS}
	For any $x,y\in\Rbb_+$, we have
	\begin{align}
	\label{eqn:nu_TC}
		\nuTC(x,y) &=
		\frac{\min(x,y)(\alpha^x+\alpha^y)}{\ln(\alpha)}
		+	
		\frac{2(1-\alpha^{\min(x,y)})}{\ln(\alpha)^2}, 
	\\
	\label{eqn:barnu_TC}
		\barnuTC(x)&=
		\frac{x\alpha^x\ln(\alpha)+2(1-\alpha^{x})}{\ln(\alpha)^2}, 
	\\
	\label{eqn:nu_DC}
		\nuDC(x,y) & =
		\frac{(1-\gamma^{-\min(x,y)})\big((\alpha\gamma)^x+(\alpha\gamma)^y\big)}{\ln(\gamma)\ln(\alpha\gamma)} \notag
		\\&\qquad +	
		\frac{2-2\alpha^{\min(x,y)}}{\ln(\alpha)\ln(\alpha\gamma)}, 
	\\
	\label{eqn:barnu_DC}
		\barnuDC(x)&=
		\frac{\alpha^x\gamma^x-\alpha^x}
		{\ln(\gamma)\ln(\alpha\gamma)}
		+
		\frac{2-2\alpha^{x}}
		{\ln(\alpha)\ln(\alpha\gamma)}, 
	\\
	\label{eqn:nu_SS}
	\nuSS(x,y)&=
			\frac{(\alpha^{\min(x,y)}-1)(\alpha^{2x}+\alpha^{2y})}{2\ln(\alpha)^2}
			\notag
			\\&  
			-\frac{\min(x,y)(\alpha^{3x}+\alpha^{3y})}{9\ln(\alpha)}
			+	
			\frac{7-7\alpha^{3\min(x,y)}}{27\ln(\alpha)^2},
	\\
	\label{eqn:barnu_SS}
	\barnuSS(x)&=
	\frac{14-27\alpha^{2x} +13\alpha^{3x}}{54\ln(\alpha)^2}
	-\frac{x\alpha^{3x}}{9\ln(\alpha)}.	
	\end{align}
\end{theorem}
\begin{proof}
See Appendix \ref{sec:appendix_proof_nu_TC_DC_SS}.
\end{proof}
Similar to the previous theorem, we can calculate the closed-form of $\|\varphi_0\|_{\Hk}^2=\inner{\varphi_0}{\varphi_0}_{\Hk}$ for the standard stable kernels. The next theorem presents these closed-forms.
\begin{theorem}\label{thm:norm_phi_0_CT}
We have $\|\varphi_{\TC,0}\|^2 =  2(\ln(\alpha))^{-2}$, $\|\varphi_{\DC,0}\|^2 =   2(\ln(\alpha)\ln(\alpha\gamma))^{-2}$, and
$\|\varphi_{\SS,0}\|^2 =  7(27\ln(\alpha))^{-2}$.
\end{theorem}
\begin{proof}
See Appendix \ref{sec:appendix_proof_norm_phi_0_CT}.
\end{proof}	
Based on the above discussion, we can derive the key elements of optimization \eqref{eqn:opt_dc_gain_2} and solve Problem \ref{prob:ID_w_DC-gain_SI} to estimate the impulse response of system.  The outline of this procedure is summarized in Algorithm \ref{alg:SSG_ID}.
\begin{algorithm}[t]
	\caption{System Identification with Steady-State Gain Side-Information}\label{alg:SSG_ID}
	\begin{algorithmic}[1]
		\State \textbf{Input:} Set of data $\Dscr$, integrable kernel $\kernel$, index set $\Ical$, convex function $\ell$,  regularization weight $\lambda$, and, real scalar $\delta$ or interval $[\udelta,\odelta]$ for the steady-state gain.
		\State Calculate matrix $\Phi$ in \eqref{eqn:Phi}.
		\Statex $\triangleright$
		{For discrete-time case, use Theorem \ref{thm:Phi_DT} and Remark \ref{rem:Phi_DT}.}
		\Statex $\triangleright$
		{For continuous-time case and step input, use \eqref{eqn:kappa}, \eqref{eqn:barkappa}, Theorems \ref{thm:inner_phi_0i_ij_CT}, \ref{thm:nu_TC_DC_SS} and \ref{thm:norm_phi_0_CT}.}
		\State Obtain matrix $\mxA$ introduced in  \eqref{eqn:A} as a sub-matrix of $\Phi$. 
		\State Obtain vector $\vca_0$ as the first column of $\Phi$.
		\State Solve convex program \eqref{eqn:opt_dc_gain_2} to obtain $\xstar$.
		\Statex $\triangleright$
		{If the steady-state gain is known to be $\delta$ and the empirical loss is the sum of squared errors, obtain $\xstar$ by \eqref{eqn:xstar} or \eqref{eqn:QP_KKT}.}
		\State Calculate $\varphi_{0}\!$ by \eqref{eqn:phi_0t}, or by Theorems \ref{thm:phi_0t_CT} and \ref{thm:phi_0t_DT}.
		\State Calculate $\varphi_1,\ldots,\varphi_{\nI}$ based on \eqref{eqn:phiu_tau_t}. 
		\Statex $\triangleright$ For continuous-time case and step input, use Theorems \ref{thm:phiu_tau_t_general} and \ref{thm:psi_TC_DC_SS}.
		\State Given $\xstar$ and $\{\varphi_i\}_{i=0}^{\nI}$, obtain $\gstar$ based on \eqref{eqn:gstar_parametric_form}.
		\State \textbf{Output:} $\gstar$ and $\xstar$.
	\end{algorithmic}
\end{algorithm}
\subsection{Hyperparameter Tuning}\label{ssec:HP}
To employ Algorithm \ref{alg:SSG_ID}, in addition to the set of data $\Dscr$, an appropriate stable kernel $\kernel$ and the regularization weight $\lambda$ are required. The general form of the kernel depends on the shape and smoothness of the impulse response to be identified. Following specifying the type of kernel, in addition to $\lambda$, it is required to determine the hyperparameters $\theta_{\kernel}$ characterizing kernel $\kernel$. Therefore, we need to estimate the vector of hyperparameters  $\theta := [\lambda, \theta_{\kernel}]$ in the admissible set $\Theta\subseteq \Rbb^{n_{\theta}}$. For this purpose, we use a cross-validation mechanism equipped with a Bayesian optimization heuristic \cite{srinivas2012information}. More precisely, the index set of data is partitioned into disjoint sets  $\Ical_{\text{T}}$ and $\Ical_{\text{V}}$ to be utilized respectively for training and validation. The prediction error on the validation data, the model evaluation metric $v:\Theta\to\Rbb$, is defined as
\begin{equation}\label{eqn:v_model_evaluation_metric}
	v(\theta) = \frac{1}{|\Ical_{\text{V}}|}\sum_{i\in\Ical_{\text{V}}}
	\big(
	y_{t_i}-\Lu{t_i}(\vcg(\theta))
	\big)^2,
\end{equation}
where $\vcg(\theta)$ is the impulse response identified using the proposed identification technique given the training data and the hyperparameters $\theta$.
Then, $\theta$ is estimated  as $\hat{\theta}:=\argmin_{\theta\in\Theta} \!\ v(\theta)$.  One can see that the dependency of model evaluation metric $v$ to the vector of hyperparameters is in a black-box oracle form. To solve this optimization problem, we use a Bayesian optimization algorithm such as GP-LCB,  which is available through \MATLAB's \texttt{bayesopt} function \cite{srinivas2012information}.

\section{Numerical Examples}\label{sec:numerical}
In this section, we demonstrate and compare the performance of the proposed method through several numerical examples.

\begin{example}\label{exm:pf_ext}\normalfont
	For the example provided in Section \ref{sec:problem_statement}, we employ the proposed identification scheme presuming that the steady-state gain of the system is given, i.e., we know that $\ell_0(\gS)=1$.
	Accordingly, we apply Algorithm \ref{alg:SSG_ID} given the set of measurement data $\Dscr$ and $\delta=1$.
	For the choice of kernel, we utilize $\kernelTC$ introduced in \eqref{eqn:TC_kernel}.
	The hyperparameters of kernel are tuned based on the cross-validation mechanism introduced in Section \ref{ssec:HP}.
	To this end, the first $80$\% of data points are chosen for training and the remaining $20$\% for validation. 
	Given this partitioning of the measurement data, the hyperparameters $\theta=[\lambda,\alpha]$ are estimated as $\argmin_{\theta\in\Theta} \!\ v(\theta)$, where $v$ is the model evaluation metric defined in \eqref{eqn:v_model_evaluation_metric}. The solution of this optimization problem is obtained by utilizing GP-LCB Bayesian optimization approach \cite{srinivas2012information}.
	
	Figure \ref{fig:example_1} shows the step response corresponding to $\gstar$ together with the step responses of the models $\hat{\vcg}_1$ and $\hat{\vcg}_2$, estimated in Section \ref{sec:problem_statement}.
	For the estimated impulse response $\gstar$, we have $\ell_0(\gstar)=1.00$.
	To evaluate and compare quantitatively the estimated impulse responses, we employ the
	following performance metric
	\begin{equation}\label{eqn:R2}
		\mathrm{fit}(\vcg) 
		= 100 \times \left(1-\frac{\|\vcg-\gS\|_2}{\|\gS\|_2}\right),
		\quad \forall \vcg\in\Hk.
	\end{equation}
	For the estimated impulse responses, we have $\mathrm{fit}(\hat{\vcg}_1)=70.88$\%, $\mathrm{fit}(\hat{\vcg}_2)=58.60$\%, and $\mathrm{fit}(\gstar)=95.84$\%.
	Accordingly, one can see that the proposed scheme outperforms in terms of performance metric $\mathrm{fit}$ and the precision of resulting steady-state gain.
	\xqed{$\triangle\!\!$}
	\begin{figure}[t]
		\centering
		\includegraphics[width =0.45\textwidth]{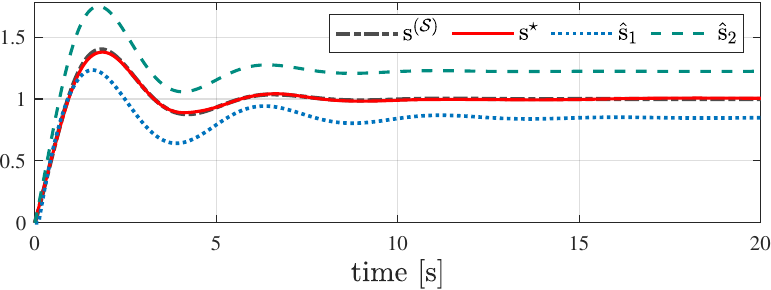}
		\caption{The step responses for system $\Scal$ and the model estimated in Example \ref{exm:pf_ext}. The impulse response $\vcs^\star$ corresponds to the proposed method.}
		\label{fig:example_1}
	\end{figure}
\end{example}
\begin{example}\label{exm:MC}\normalfont
In this example, we compare the performance of the proposed identification method  with  the existing schemes through a Monte Carlo analysis.
To this end, with respect to each $(n,r)$ in $\{16,17,\ldots,25\}\times\{0.8,0.82,\ldots,0.96\}$, we employ \MATLAB's \texttt{drss} function to randomly generate a discrete-time LTI system with order $n$ and spectral radius $r$.
We normalize these systems with their $\Hcal_2$-norm and set them initially at rest.
For each of these systems, a random zero-mean white Gaussian input signal with length  $\nD=200$ is generated using \MATLAB's \texttt{idinput} function.
By applying these input signals to the systems, we obtain their noiseless output signals.
We consider three levels of signal-to-noise ratio (SNR): high, medium, and low, 
which are $5\ \!$dB, $15\ \!$dB, and $25\ \!$dB, respectively. 
For the additive measurement uncertainty, we generate a zero-mean white Gaussian signal for each of these SNR levels and each output signal.
The noiseless output signal is then corrupted with the the corresponding additive noise signals, and the resulting noisy output is measured at time instants $t_i=i$, for $i=0,1,\ldots,199$.
As a result, we have $100$ sets of input-output data for each of the aforementioned SNR values as following
\begin{equation*}
	\begin{split}
		\Dscr_i^{\supscrpsm{5\text{dB}}}
		&\!=\!
		\big\{(u_s^{\supscrpsm{i}},y_s^{\supscrpsm{i,5\text{dB}}})\big|s \!=\! 0,\ldots,199\big\},\ \ i \!=\! 1,\ldots,100,	
		\\
		\Dscr_i^{\supscrpsm{15\text{dB}}}
		&\!=\!
		\big\{(u_s^{\supscrpsm{i}},y_s^{\supscrpsm{i,15\text{dB}}})\big|s \!=\! 0,\ldots,199\big\},\ \ i \!=\!1,\ldots,100,	
		\\
		\Dscr_i^{\supscrpsm{25\text{dB}}}
		&\!=\!
		\big\{(u_s^{\supscrpsm{i}},y_s^{\supscrpsm{i,25\text{dB}}})\big|s \!=\! 0,\ldots,199\big\},\ \ i \!=\!1,\ldots,100,	
	\end{split}
\end{equation*}	
where the superscript indicates the SNR value in the respective data set.
We employ the above input-output data sets and the following identification methods to estimate the impulse response of corresponding systems:
\begin{itemize}
	\item[A.] This method is a modified subspace approach incorporating steady-state features of output \cite{yoshimura2019system}. 
	\item[B.] This method estimates impulse response by solving a constrained optimization problem formulated based on a behavioral approach \cite{markovsky2017subspace}.
	\item[C.] This method considers the interpretation of subspace identification as the optimal multi-step ahead prediction and modifies it to a constrained least-squares problem where the imposed equality constraint models approximately the steady-state gain information \cite{alenany2011improved}.
	\item[D.] 
	This method is a general Bayesian variant of the optimal multi-step ahead predictor interpretation of subspace identification approach, where steady-state gain information is integrated into the covariance of the prior distribution	
	\cite{trnka2009subspace}. 
	\item[E.] 
	In this  method, the step response of system is first  estimated by a kernel-based Bayesian approach, 
	and then, 
	the FIR is calculated using discrete derivative \cite{fujimoto2018kernel}.
	\item[F.] 
	This  method estimates a FIR model for the system based on a kernel-based Bayesian approach where the steady-state gain information is enforced on the total summation of the estimated FIR \cite{tan2020kernel}.	 
	\item[G.] 
	The last method is the scheme proposed in this paper and summarized in Algorithm \ref{alg:SSG_ID}.	
\end{itemize}
The kernel-based methods E, F, and G employ the same kernel type \eqref{eqn:DC_kernel} to give a fair comparison.
To evaluate and compare the estimation performances of these methods, we employ the {R-squared} metric introduced in \eqref{eqn:R2}.
\begin{figure}[t]
	\begin{center}
		\includegraphics[width=0.35\textwidth]{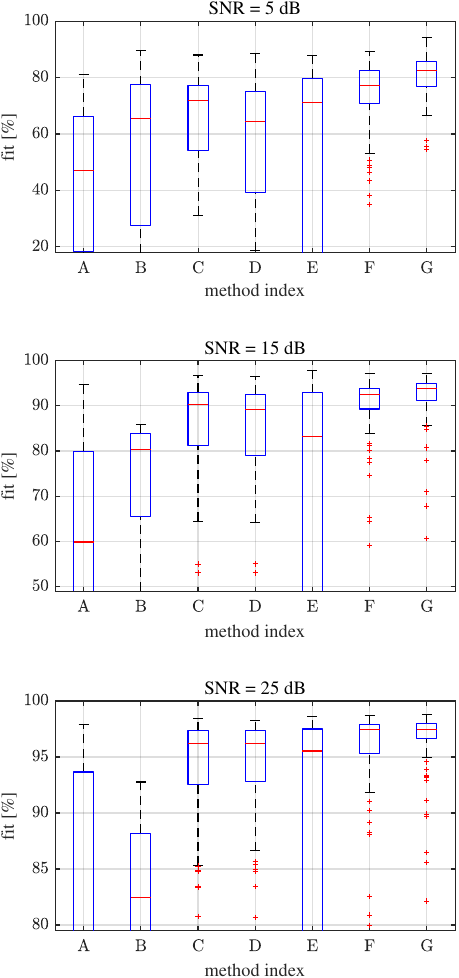}
	\end{center}
	\caption{
		The box-plots compare methods discussed in Example \ref{exm:MC} via fitting metric \eqref{eqn:R2} for different SNR levels. One can see the proposed method (G) shows superior performance. The y-axis range has been adjusted to improve visibility in each of these box plots.
	}
	\label{fig:exm1_box_plot}
\end{figure}
Figure \ref{fig:exm1_box_plot} demonstrates and compares the values of R-squared metric for the estimation results of the above mentioned methods and SNR levels.

\noindent\underline{\textbf{Discussion:}}
From Figure \ref{fig:exm1_box_plot}, we observe that the proposed identification scheme demonstrates better estimation performance than other methods.
Indeed, methods A, B, C, and D are based on the subspace approach and prediction error minimization. Meanwhile, the proposed scheme is a kernel-based method in which the stability of the system is included and the model complexity tuning is implemented based on the effective approach of estimating continuous regularization hyperparameters rather than selecting an integer order \cite{pillonetto2014kernel,ljung2020shift}. 
The methods E and F estimate a FIR model for the system, which can be inexact when the spectral radius of the system is close to one, and the impulse response of system can not be approximated well by a short FIR. Moreover, the discrete derivative employed in method F makes the estimation prone to numerical sensitivity, particularly when the data is noisy. On the other hand, the proposed scheme estimates directly  the impulse response without any truncation and inexact numerical procedures such as discrete derivatives.
\xqed{$\triangle$}
\end{example}

\section{Conclusion}
In this work, we have addressed the impulse response identification problem when side-information on the steady-state gain of the system is provided. The problem is formulated as a generic nonparametric identification in the form of an infinite-dimensional constrained convex program over the RKHS of stable impulse responses. This optimization problem is designed such that the objective function corresponds to the regularized empirical loss, and the imposed linear constraints enforce the integration of the given side-information into the solution. We have shown that the steady-state gain is a bounded operator over the employed RKHS, which results in guaranteeing the existence and uniqueness of the solution.  By using the representer theorem, the optimization problem is reduced to a  finite-dimensional convex program, which can be solved efficiently. In the case of exact side-information, quadratic empirical loss, and quadratic regularization, the identification problem has a closed-form solution. 
Compared to the existing methods, the proposed identification approach can utilize non-uniform measurement samples and integrate steady-state gain side-information into the direct approach of continuous-time system identification. 
Moreover, through an extensive Monte Carlo numerical experiment, we have verified that the introduced methodology outperforms the benchmark approaches. 
The proposed scheme has several features which have led to the observed superior performance, including direct estimation of the impulse response without any truncation and 
without inexact numerical procedures such as discrete derivatives.  The method uses 
kernel-based regularization to enforce the BIBO stability of the estimated impulse response, and, effective model selection and complexity tuning through estimating continuous variables such as regularization weight and hyperparameters.

\appendix
\section{Appendix} 
\label{sec:Appendix}
\subsection{Integrability of the Standard Kernels}
\label{sec:appendix_proof_kTC_kDC_kSS_integrable}
For a Mercer kernel $\kernel$, let assume that there exist $C\in\Rbb_+$ and $\alpha\in(0,1)$ such that $|\kernel(s,t)|\le C\alpha^{\frac12(s+t)}$, for any $s,t\in\Tbb$. 
Subsequently, when $\Tbb=\Zbb_+$, we have
\begin{equation}
	\begin{split}
		\sum_{s,t\ge 0}|\kernel(s,t)|
		&\le
		C\sum_{s\ge 0}\sum_{t\ge 0}\alpha^{\frac12(s+t)}
		\\&
		=
		C\sum_{s\ge 0}\alpha^{\frac{1}2s}\sum_{t\ge 0}\alpha^{\frac{1}2t}
		\\&
		=
		\frac{C}{(1-\alpha^{\frac{1}2})^2}
		<\infty.
	\end{split}
\end{equation}
Similarly, when $\Tbb=\Rbb_+$, one has
\begin{equation}
	\begin{split}\!\!\!\!
		\int_{\Rbb_+\times\Rbb_+}
		|\kernel(s,t)|\drm s \drm t
		& \le
		C\int_{\Rbb_+}\!\int_{\Rbb_+}\!
		\alpha^{\frac12(s+t)}\drm s \drm t
		\\&=
		C\int_{\Rbb_+}\! \alpha^{\frac12s}\drm s\int_{\Rbb_+}\!
		\alpha^{\frac12t} \drm t
		\\&
		=\frac{4C}{(\ln(\alpha))^2}<\infty.
	\end{split}
\end{equation}
Therefore, $\kernel$ is an integrable kernel.
For any $s,t\in\Tbb$, one can easily see that 
$|\kernelTC(s,t)|\le \alpha^{\frac12(s+t)}$, 
$|\kernelDC(s,t)|\le \alpha^{\frac12(s+t)}$, and
\begin{equation}
	\begin{split}
		\!\!\!\!\!\!
		|\kernelSS(s,t)| &= \Big|\alpha^{\max(s,t)}\Big[\alpha^{s+t}-\frac{1}{3}\alpha^{2\max(s,t)}\Big]\Big|
		\\&
		\le
		\Big[|\alpha^{s+t}|+\frac{1}{3}|\alpha^{2\max(s,t)}|\Big]\alpha^{\max(s,t)}
		\\&
		\le 
		\frac43 \alpha^{\frac12(s+t)}.
	\end{split}
\end{equation}
According to our previous discussion, this implies that $\kernelTC$, $\kernelDC$ and $\kernelSS$ are integrable kernels.
\qed
\subsection{Proof of Theorem~\ref{thm:dc_bounded}} \label{sec:appendix_proof_dc_bounded}
\textbf{Case $\Tbb=\Zbb_+$:}
For each $n\in\Zbb_+$, define $\vcf_n=(f_{n,s})_{s\in\Zbb_+}$ as
$\vcf_n=\sum_{t=0}^n\kernel_t$.
Since,  for each $t\in\Zbb_+$, one has $\kernel_t\in\Hk$, we know that $\vcf_n\in\Hk$. Furthermore, from the reproducing property, one can see that 
\begin{equation}
	\|\vcf_n\|_{\Hk}^2=\sum_{s=0}^n\sum_{t=0}^n\kernel(s,t).
\end{equation}
Since $\kernel$ is an integrable kernel, we know that $\sum_{s,t\ge 0}|\kernel(s,t)|< \infty$. Accordingly, for any positive real scalar $\varepsilon$, there exists $N\in\Zbb_+$ such that 
$\sum_{s,t\ge N}|\kernel(s,t)|\le \varepsilon^2$.
Let $N_{\varepsilon}$ denote smallest non-negative integer with this property.
For any $n,m\in\Zbb_+$ such that $n> m\ge N_{\varepsilon}$, we have 
$\vcf_n-\vcf_m = \sum_{t=m+1}^{n}\!\kernel_t$.
Accordingly, from the reproducing property, one can see that
\begin{equation}
	\|\vcf_n-\vcf_m\|_{\Hk}^2 = \sum_{s = m+1}^n \sum_{t=m+1}^n\kernel(s,t).
\end{equation}
Subsequently, from $n,m\ge N_{\varepsilon}$, the triangle inequality and the definition of $N_{\varepsilon}$, it follows that $\|\vcf_n-\vcf_m\|_{\Hk}\le \varepsilon$.
Therefore, $\{\vcf_n\}_{n\ge 0}$ is a Cauchy sequence in Hilbert space $\Hk$, and consequently, we know that there exists $\vcf=(f_s)_{s\in\Zbb_+}\in\Hk$ such that $\lim_{n\to\infty}\|\vcf_n-\vcf\|_{\Hk}=0$.
For any $s\in\Zbb_+$, from the reproducing property, we know that $f_s-f_{n,s}=\inner{\vcf-\vcf_n}{\kernel_s}_{\Hk}$.
Accordingly,  due to the Cauchy-Schwartz inequality, we have 
\begin{equation}
	\lim_{n\to\infty}|f_s-f_{n,s}| 
	\le 	
	\lim_{n\to\infty}\|\vcf-\vcf_n\|_{\Hk}\!\ \|\kernel_s\|_{\Hk} = 0,
\end{equation}
i.e., one has $\lim_{n\to\infty}f_{n,s}=f_s$, for any $s\in\Zbb_+$, which
implies that $\vcf=\sum_{t\in\Zbb_+}\kernel_t$. Therefore, $\vcf$ coincides with $\varphi_{0}$ defined by \eqref{eqn:phi_0t}.
Moreover, from $\lim_{n\to\infty}\vcf_{n}=\vcf$, the dominated convergence theorem and being $\kernel$ an integrable kernel, we have
\begin{equation*}
	\begin{split}\!\!
		\big\|\vcf\big\|_{\Hk}^2
		= \!
		\lim_{n\to\infty}\!\big\|\vcf_n\big\|_{\Hk}^2
		= 
		\lim_{n\to\infty}\!\sum_{0\le s,t\le n}\!\!\!\kernel(s,t)
		= \!\!\!
		\sum_{s,t\in\Zbb_+}\!\!\kernel(s,t),
	\end{split}
\end{equation*}  
which implies \eqref{eqn:norm_phi_0}.
For any $\vcg=(g_t)_{t\in\Zbb_+}\in\Hk$, we know that $\vcg$ is integrable, i.e., $\sum_{t\in\Zbb_+}|g_t|<\infty$.
Therefore, 
from $\lim_{n\to\infty}\vcf_n = \vcf$ and the reproducing property, 
it follows that
\begin{equation}
	\begin{split}
		\sum_{t\in\Zbb_+} g_t 		 
		&= 
		\lim_{n\to\infty}\sum_{0\le t\le n} g_t 
		\\&
		=\lim_{n\to\infty}\Big\langle\sum_{0\le t\le n}\!\!\kernel(\cdot,t),\vcg\Big\rangle_{\Hk}
		\\&
		= \lim_{n\to\infty}\inner{\vcf_n}{\vcg}_{\Hk}
		= \inner{\vcf}{\vcg}_{\Hk},
	\end{split}
\end{equation}
which implies \eqref{eqn:inner_phi_0_g} for the case of $\Tbb=\Zbb_+$.

\textbf{Case $\Tbb=\Rbb_+$:}
Let $r\in\Rbb_+$ and 
$I_r:=\int_0^r\!\int_0^r\!\kernel(s,t)\drm s\drm t$, which is well-defined since $\kernel$ is an integrable kernel.
Define $\vcf^\rth_n:=(f_{n,s}^\rth)_{s\in\Rbb_+}\in\Hk$ as
\begin{equation}\label{eqn:def_fn}
f^\rth_{n,s}=\frac{r}n\sum_{i=0}^{n-1}\kernel(\frac{i}{n}r,s), \qquad \forall s\in\Rbb_+,
\end{equation}
for each $n\in\Nbb$.
From the reproducing property, for any $n,m\in\Zbb_+$, we know that
\begin{equation*}
\begin{split}
	\Big\langle\vcf_n^\rth&,\vcf_m^\rth\Big\rangle_{\Hk}-I_r \\&=
	\Big\langle\frac{r}{n}\!\sum_{i=0}^{n-1}\kernel(\frac{i}{n}r,\cdot),
	\frac{r}{m}\!\sum_{j=0}^{m-1}\kernel(\frac{j}{m}r,\cdot)\Big\rangle_{\Hk}-I_r
	\\&=
	\sum_{i=0}^{n-1}\sum_{j=0}^{m-1}\bigg[\kernel(\frac{i}{n}r,\frac{j}{m}r)\frac{r^2}{nm}
	-
	\int_{\frac{i}{n}r}^{\frac{i+1}{n}r}\!\!\int_{\frac{j}{m}r}^{\frac{j+1}{m}r}\!\!
	\kernel(s,t)\drm s \drm t\bigg]
	\\&=
	\sum_{i=0}^{n-1}\sum_{j=0}^{m-1}
	\int_{\frac{i}{n}r}^{\frac{i+1}{n}r}\!\!\int_{\frac{j}{m}r}^{\frac{j+1}{m}r}
	\bigg[\kernel(\frac{i}{n}r,\frac{j}{m}r)-\kernel(s,t)\bigg]\drm s \drm t.
\end{split}
\end{equation*} 
Since $[0,r]^2$ is a compact region, continuity of $\kernel$ implies that $\kernel$ is uniformly continuous on $[0,r]^2$.
Therefore, for any $\varepsilon>0$, there exists ${\delta}_{\varepsilon}\in\Rbb_+$ such that we have 
\begin{equation}\label{eqn:k_s1t1_k_s2t2}
|\kernel(s_1,t_1)-\kernel(s_2,t_2)|\le \frac{\varepsilon^2}{4r^2},
\end{equation} 
for any $s_1,s_2,t_1,t_2\in [0,r]$ where $|s_1-s_2|+|t_1-t_2|\le {\delta}_{\varepsilon}$.
Accordingly, if $n,m\ge n_{\varepsilon}$, where $n_{\varepsilon}$ is the smallest positive integer larger than 
$\frac{1}{\delta_{\varepsilon}}2r$, one has
\begin{equation*}
|\inner{\vcf_n^\rth}{\vcf_m^\rth}_{\Hk}-I_r| 
\le
\sum_{i=0}^{n-1}\sum_{j=0}^{m-1}
\int_{\frac{i}{n}r}^{\frac{i+1}{n}r}\!\!\int_{\frac{j}{m}r}^{\frac{j+1}{m}r}
\!\frac{\varepsilon^2}{4r^2}\!\ \drm s \drm t=
\frac{1}{4}\varepsilon^2.
\end{equation*}
Therefore, we know that 
\begin{equation}\label{eqn:I_inner_fn_fm_eps}
I_r-\frac{1}{4}\varepsilon^2 \le \inner{\vcf_n^\rth}{\vcf_m^\rth}_{\Hk}\le I_r+\frac{1}{4}\varepsilon^2, \quad \forall n,m\ge n_{\varepsilon}.
\end{equation}
Subsequently, one can see that	
\begin{equation*}
\begin{split}
	\|\vcf_n^\rth-\vcf_m^\rth\|^2 
	&=
	\inner{\vcf_n^\rth}{\vcf_n^\rth}_{\Hk}
	-2\inner{\vcf_n^\rth}{\vcf_m^\rth}_{\Hk} + \inner{\vcf_m^\rth}{\vcf_m^\rth}_{\Hk}
	\\&\le
	(I + \frac{1}{4}\varepsilon^2)
	-2(I-\frac{1}{4}\varepsilon^2) +
	(I + \frac{1}{4} \varepsilon^2) 
	= 
	\varepsilon^2,
\end{split}
\end{equation*} 
for any $n,m\ge n_{\varepsilon}$. 
Accordingly,  $\{\vcf_n^\rth\}_{n=1}^{\infty}$ is a Cauchy sequence in the Hilbert space $\Hk$, which implies that there exists $\vcf^\rth=(f_s^\rth)_{s\in\Rbb_+}$ in $\Hk$ such that $\{\vcf_n^\rth\}_{n=1}^{\infty}$ converges to $\vcf^\rth$.
Due to the reproducing property, we know that $f_s^\rth-f_{n,s}^\rth=\inner{\vcf^\rth-\vcf_n^\rth}{\kernel_s}_{\Hk}$, for any $s\in\Rbb_+$.
Consequently,  from the Cauchy-Schwartz inequality, it follows that
\begin{equation}
\lim_{n\to\infty}|f_s^\rth-f_{n,s}^\rth| 
\le 	
\lim_{n\to\infty}\|\vcf^\rth-\vcf_n^\rth\|_{\Hk}\!\ \|\kernel_s\|_{\Hk} = 0,
\end{equation}
i.e., we have $\lim_{n\to\infty}f_{n,s}^\rth=f_s^\rth$.
On the other hand, according to \eqref{eqn:k_s1t1_k_s2t2}, one can see that
\begin{equation*}
\begin{split}
	|f_{n,s}^\rth\!\!\!\!&\ \ -\int_0^r\!\kernel(s,t)\drm t| 
	=
	\bigg|\sum_{i=0}^{n-1}
	\int_{\frac{i}{n}r}^{\frac{i+1}{n}r}\!\Big[
	\kernel(\frac{i}{n}r,s)-\kernel(t,s)\Big] \drm t
	\bigg|
	\\ 
	&\!\! \!\! \le 
	\sum_{i=0}^{n-1}
	\int_{\frac{i}{n}r}^{\frac{i+1}{n}r}\!\Big|
	\kernel(\frac{i}{n}r,s)-\kernel(t,s)\Big| \drm t
	\\&\!\! \!\! \le 
	\sum_{i=0}^{n-1}
	\frac{1}n \frac{\varepsilon^2}{4r^2} 
	= 
	\frac{\varepsilon^2}{4r^2}.
\end{split}
\end{equation*}
Subsequently, we have $f_s^\rth=\lim_{n\to\infty}f_{n,s}^\rth = \int_0^r\kernel(s,t)\drm t$, i.e.,
$\vcf^\rth=\int_0^r\kernel(\cdot,t)\drm t$.
Accordingly, from $\lim_{n\to\infty}\vcf_n^\rth=\vcf^\rth$, \eqref{eqn:I_inner_fn_fm_eps} and the definition of $I_r$, it follows that
\begin{equation}\label{eqn:norm_fr}
\begin{split}
	\Big\|\int_0^r &\!\kernel(\cdot,t)\drm t\Big\|_{\Hk}^2
	= \|\vcf^\rth\|_{\Hk}^2 \\
	&= \limOp_{n\to\infty}\inner{\vcf_n^\rth}{\vcf_n^\rth}_{\Hk}
	= \int_0^r \!\int_0^r \!\kernel(s,t)\drm s\drm t.
\end{split}
\end{equation} 
Take an arbitrary $\vcg=(g_t)_{t\in\Rbb_+}$ in $\Hk$ and $t,\varepsilon\in\Rbb_+$ such that $t+\varepsilon\in\Rbb_+$.
Since $\kernel$ is symmetric and continuous, and due to the Cauchy-Schwartz inequality and the reproducing property, we have
\begin{equation*}
\begin{split}
	\lim_{\varepsilon\to 0}&|g_{t+\varepsilon}-g_t|
	=
	\lim_{\varepsilon\to 0}|\inner{\kernel_{t+\varepsilon}-\kernel_t}{\vcg}_{\Hk}|
	\\& \!\! \!\! \!\! \!\!  \!\! 
	\le 
	\lim_{\varepsilon\to 0}\|\kernel_{t+\varepsilon}-\kernel_t\|_{\Hk} \|\vcg\|_{\Hk}\\
	& \!\! \!\! \!\! \!\!  \!\! 
	= 
	\lim_{\varepsilon\to 0}
	\big[\kernel(t+\varepsilon,t+\varepsilon)-2\kernel(t,t+\varepsilon)+\kernel(t,t)\big]^{\frac12}
	\|\vcg\|_{\Hk} = 0,
\end{split}
\end{equation*}
which implies the continuity of $\vcg=(g_t)_{t\in\Rbb_+}$ as a function of $t$.
Therefore, the Riemann integral of $\vcg$ on $[0,r]$ is well-defined.
Accordingly, from the definition of $\vcf_n^\rth$, the reproducing property  and $\lim_{n\to\infty}\vcf_n^\rth = \vcf^\rth$,
it follows that 
\begin{equation*}
\begin{split}
	\int_0^r \!  g_t\drm t 
	&= 
	\lim_{n\to\infty}\frac{r}n\sum_{i=0}^{n-1}g(\frac{i}nr)
	\\&=
	\lim_{n\to\infty}\inner{\vcf_n^\rth}{\vcg}_{\Hk}
	\\&=
	\inner{\vcf^\rth}{\vcg}_{\Hk}
	= 
	\Big\langle\int_0^r \!\kernel(\cdot,t)\drm t,\vcg\Big\rangle_{\Hk}.
\end{split}
\end{equation*}
With respect to each $n\in\Nbb$, let $r_n$ be
\begin{equation*}
r_n:=\min\Big\{r\ge n
\Big|
\int_r^{\infty}\!\int_r^{\infty}\!|\kernel(s,t)|\drm s\drm t\le \frac1n\Big\},
\end{equation*}
which is well-defined since $\kernel$ is continuous and integrable.
Moreover, one can see that $\{r_n\}_{n=1}^{\infty}$ is an unbounded strictly increasing sequence.
For any $n,m\in\Nbb$, such that $n\le m$, we have 
\begin{equation}
\begin{split}
f_t^\rthn-f_t^\rthm	
&=
\int_0^{r_n}\!\kernel(s,t)\drm s
-
\int_0^{r_m}\!\kernel(s,t)\drm s
\\&=
\int_{r_m}^{r_n}\!\kernel(s,t)\drm s.
\end{split}	
\end{equation}
Accordingly, one can see that
\begin{equation*}
\begin{split}
	\|\vcf^\rthn-\vcf^\rthm\|_{\Hk}^2 
	&=
	\inner{\vcf^\rthn}{\vcf^\rthn-\vcf^\rthm}_{\Hk}
	\!-\!
	\inner{\vcf^\rthm}{\vcf^\rthn-\vcf^\rthm}_{\Hk}
	\\&=
	\int_0^{r_n}\!\! f_t^\rthn-f_t^\rthm\drm t
	-
	\int_0^{r_m}\!\! f_t^\rthn-f_t^\rthm\drm t
	\\&=
	\int_{r_m}^{r_n}\!\int_{r_m}^{r_n}\!
	\kernel(s,t)\drm s\drm t
	\\&\le \int_{r_m}^{r_n}\!\int_{r_m}^{r_n}\!
	|\kernel(s,t)|\drm s\drm t
	\le \frac{1}{m}.
\end{split}
\end{equation*}
This implies that $\{\vcf^\rthn\}_{n=1}^\infty$ is a Cauchy sequence in $\Hk$.
Therefore, we know that there exists $\vcf=(f_s)_{s\in\Rbb_+}$ in $\Hk$ such that $\lim_{n\to\infty}\|\vcf-\vcf^\rthn\|_{\Hk}=0$. 
Subsequently, due to the reproducing property, for any $s\in\Rbb_+$, we have
\begin{equation*}
\begin{split}
	f_s =\inner{\vcf}{\kernel_{s}}_{\Hk} 
	&= \lim_{n\to\infty}\inner{\vcf^\rthn}{\kernel_{s}}_{\Hk}
	\\&= 
	\lim_{n\to\infty}
	\int_{0}^{r_n}\!
	\kernel(s,t)\drm t
	= 
	\int_{0}^{\infty}\!
	\kernel(s,t)\drm t,
\end{split}
\end{equation*}
where the last equality holds due to $\kernel_s\in\Hk$, $\Hk\subset\Lscrone$, and the dominated convergence theorem.
Note that $\vcf$ coincides with $\varphi_{0}$ defined by \eqref{eqn:phi_0t}.
Accordingly, due to $\lim_{n\to\infty}\vcf^\rthn=\vcf$, \eqref{eqn:norm_fr}, the definition of $\vcf$, and the dominated convergence theorem, we have
\begin{equation*}
\begin{split}
	\Big\|\int_0^\infty \kernel(\cdot,t)\drm t\Big\|_{\Hk}^2
	&= \|\vcf\|_{\Hk}^2 
	\\&= \lim_{n\to\infty}\|\vcf^\rthn\|_{\Hk}^2
	\\&= 
	\lim_{n\to\infty} \int_0^{r_n} \!\int_0^{r_n} \!\kernel(s,t)\drm s\drm t
	\\&=
	\int_0^{\infty} \!\int_0^{\infty} \!\kernel(s,t)\drm s\drm t.
\end{split}
\end{equation*}
Based on similar arguments, for any $\vcg=(g_t)_{t\in\Rbb_+}\in\Hk$, one can see that
\begin{equation*}
\begin{split}
	\Big\langle\int_0^\infty\!\kernel(\cdot,t)\drm t,\vcg\Big\rangle_{\Hk}
	&= 
	\inner{\vcf}{\vcg}_{\Hk}
	\\&=
	\lim_{n\to\infty}\inner{\vcf^\rthn}{\vcg}_{\Hk}
	\\&=
	\lim_{n\to\infty} 
	\Big\langle\int_0^{r_n}\kernel(\cdot,t)\drm t,\vcg\Big\rangle_{\Hk}
	\\&= \lim_{n\to\infty}\int_0^{r_n}\!\! g_t\drm t
	= \int_0^\infty\!\! g_t\drm t,
\end{split}
\end{equation*}  
where the last equality is due to the dominated convergence theorem and the fact that $\vcg$ is integrable.
Therefore, we have  $\ell_0(\vcg)=\inner{\vcf}{\vcg}_{\Hk}$, 
which implies \eqref{eqn:inner_phi_0_g}. This concludes the proof.
\qed
\subsection{Proof of Theorem~\ref{thm:G_nonempty_closed_convex}} \label{sec:appendix_proof_G_nonempty_closed_convex}
Let $\delta$ be a real scalar such that $\udelta\le\delta\le\odelta$. 
We know that $\dc(\kernel_\tau)=\int_{\Rbb_+}\kernel(\tau,t)\drm t$ is non-zero.
Define $\vch$ as 
\begin{equation}\label{eqn:h_in_G}
	\vch
	=
	\frac{\delta}{\dc(\kernel_\tau)} \kernel_\tau	
	\in\Hk.
\end{equation}
One can see that $\dc(\vch)=\delta\in[\udelta,\odelta]$, and hence, $\vch\in\Gscr_{\kernel}([\udelta,\odelta])$.
Thus, $\Gscr_{\kernel}([\udelta,\odelta])$ is non-empty.
From Theorem~\ref{thm:dc_bounded}, we have
\begin{equation}\label{eqn:G_union}
	\begin{split}
		\Gscr_{\kernel}([\udelta,\odelta])
		&=	
		\big\{\vcg\in\Hk\big| \inner{\varphi_0}{\vcg}_{\Hk}\ge\udelta\big\}
		\\&\qquad\qquad
		\scalebox{1.5}{$\cap$}\!\ 
		\big\{\vcg\in\Hk\big|  \inner{\varphi_0}{\vcg}_{\Hk}\le\odelta\big\}.
	\end{split}	
\end{equation}
In the right-hand side of \eqref{eqn:G_union}, each of the sets is convex and closed.
Therefore, $\Gscr_{\kernel}([\udelta,\odelta])$ is a convex and closed subset of $\Hk$ as well. This concludes the proof.
\qed
\subsection{Proof of Lemma~\ref{thm:Lu_bounded}} \label{sec:appendix_proof_Lu_bounded}
The first part of theorem can be obtained directly from the Riesz representation theorem \cite{brezis2011functional}.  
Due to the reproducing property of kernel, for each $\tau\in\Tscr$, we have
\begin{equation}
	\phiu{\tau,t} = \inner{\phiu{\tau}}{\kernel_t}_{\Hk}=\Lu{\tau}(\kernel_t),\quad \forall t\in\Tbb.	
\end{equation}
Subsequently, one can see \eqref{eqn:phiu_tau_t} holds due to the definition of operator $\Lu{\tau}$. This concludes the proof.
\qed
\subsection{Proof of Theorem~\ref{thm:main_thm}} \label{sec:appendix_proof_main_thm}
Let $\Jcal:\Hk\to\Rbb\cup\{+\infty\}$ be defined as
\begin{equation}\label{eqn:J}
	\Jcal(\vc{g}) 
	:= 
	\Ecal_{\ell}(\vc{g}) + \lambda \Rcal(\vcg) 
	+ 
	\delta_{\Gscr_{\kernel}([\udelta,\odelta])}(\vc{g}), 
	\quad\forall \vc{g}\in\Hk.
\end{equation}
One can see that $\min_{g\in\Hk}\Jcal(g)$ is equivalent to \eqref{eqn:opt_dc_gain_1}.
For $\vch$ introduced in \eqref{eqn:h_in_G}, we know $\vch \in \Gscr_{\kernel}([\udelta,\odelta])$. Therefore, one can see that $\delta_{\Gscr_{\kernel}([\udelta,\odelta])}(\vch)=0$, and subsequently,  we have
\begin{equation}
	0
	\le
	\Jcal(\vch)
	=
	\Ecal_{\ell}(\vch)+\lambda\rho(\big\|\vch\big\|_{\Hk}) 
	< \infty,
\end{equation}
i.e., function $\Jcal$ is proper.
Due to Theorem \ref{thm:G_nonempty_closed_convex},  $\Gscr_{\kernel}([\udelta,\odelta])$ is a convex and closed subset of $\Hk$. Accordingly, $\delta_{\Gscr_{\kernel}([\udelta,\odelta])}:\Hk\to\Rbb\cup\{+\infty\}$ is a proper lower semi-continuous convex function \cite{peypouquet2015convex}. 
Moreover, from Lemma \ref{thm:Lu_bounded}, and also
the continuity and the convexity of the function $\ell$, it follows that $\Ecal_{\ell}:\Hk\to\Rbb_+$ is convex and continuous.
The convexity of function $\rho:\Rbb_+\to\Rbb_+$ implies its continuity. Furthermore, since $\rho$ is strictly increasing, we know that  $\Rcal:\Hk\to\Rbb_+$ is a strictly convex  continuous function.
Thus, $\Jcal:\Hk\to\Rbb\cup\{+\infty\}$
is a proper lower semi-continuous strictly convex function, and subsequently, the optimization problem $\min_{g\in\Hk}\Jcal(g)$ admits a unique finite solution \cite{peypouquet2015convex}, denoted by $\gstar$.
From the definition of $\Gscr_{\kernel}([\udelta,\odelta])$, Theorem \ref{thm:dc_bounded}, one has
\begin{equation}\label{eqn:delta_Gk_z}
	\delta_{\Gscr_{\kernel}([\udelta,\odelta])}(\vc{g}) = \delta_{[\udelta,\odelta]}(\inner{\varphi_0}{\vc{g}}_{\Hk}), \qquad \forall\vcg\in\Hk.	
\end{equation}
Also, according to Lemma \ref{thm:Lu_bounded}, we know that
\begin{equation}\label{eqn:Ecal_ell_z}
	\Ecal_{\ell}(\vc{g}) 
	=
	\ell([\inner{\varphi_i}{\vcg}_{\Hk}]_{i=1}^{\nI},\vcy_{\Ical}), \qquad \forall\vcg\in\Hk.	
\end{equation}
Let function $e:\Rbb^{\nI+1}\to\Rbb_+\cup\{+\infty\}$ be defined as 
\begin{equation}
	e(z_0,\ldots,z_{\nI})
	=
	\ell([z_i]_{i=1}^{\nI},\vcy_{\Ical})
	+
	\lambda 
	\delta_{[\udelta,\odelta]}(z_0),
\end{equation}
for any $[z_0,\ldots,z_{\nI}]^\tr \in \Rbb^{\nI+1}$.
Then, due to \eqref{eqn:delta_Gk_z} and \eqref{eqn:Ecal_ell_z}, one can see that
\begin{equation}
	e(\inner{\varphi_0}{\vcg}_{\Hk},\ldots,\inner{\varphi_{\nI}}{\vcg}_{\Hk})
	=
	\Ecal_{\ell}(\vcg)
	+
	\lambda
	\delta_{\Gscr_{\kernel}([\udelta,\odelta])}(\vc{g}),
\end{equation}
for any $\vcg\in\Hk$.
Hence, from Theorem \ref{thm:rep_thm} and since $\min_{g\in\Hk}\Jcal(g)$ admits a solution, 
it has a solution in $\linspan\{\varphi_i\}_{i=0}^{\nI}$ as well. 
According to the uniqueness of this solution, we know that $\gstar$ belongs to 
$\linspan\{\varphi_i\}_{i=0}^{\nI}$, i.e., $\gstar$ has the parametric form given in \eqref{eqn:opt_dc_gain_2}.
For each $\vcg\in\linspan\{\varphi_i\}_{i=0}^{\nI}$, we know that there exists $\vcx = [x_0,\ldots,x_{\nI}]^\tr\in\Rbb^{\nI+1}$ such that 
$\vcg=\sum_{j=0}^{\nI}x_j\varphi_j$.
This implies that
\begin{equation}\label{eqn:l0_g_x}
	\!\!\!
	\dc(\vcg) 
	=
	\Big\langle\varphi_0,\sum_{j=0}^{\nI}x_j\varphi_j\Big\rangle_{\Hk}\!
	=
	\sum_{j=0}^{\nI}\inner{\varphi_0}{\varphi_j}_{\Hk}x_j
	=
	\vca_0^\tr\vcx.
\end{equation}
Similarly, we have
\begin{equation}\label{eqn:Lu_g_x}
	\begin{split}
		\big[\Lu{t_i}(\vcg)\big]_{i\in\Ical} 
		&
		= 
		\Big[\Big\langle\varphi_i,\sum_{j=0}^{\nI}x_j\varphi_j\Big\rangle_{\Hk}\Big]_{i=1}^{\nI}
		\\&=
		\Big[\sum_{j=0}^{\nI}\inner{\varphi_i}{\varphi_j}_{\Hk}x_j\Big]_{i=1}^{\nI}
		= \mxA\vcx. 	
	\end{split}
\end{equation}
Moreover, from the definition of matrix $\Phi$, one can see that
\begin{equation}\label{eqn:norm_g_x}
	\begin{split}
		\|\vcg\|_{\Hk}^2
		&= 
		\Big\langle\sum_{i=0}^{\nI}x_i\varphi_i,\sum_{j=0}^{\nI}x_j\varphi_j\Big\rangle_{\Hk}
		\\&=
		\sum_{i,j=0}^{\nI}
		x_i \inner{\varphi_i}{\varphi_j}_{\Hk}x_j
		=
		\vcx^\tr\Phi\vcx.
	\end{split}
\end{equation}
Accordingly, due to \eqref{eqn:l0_g_x}, \eqref{eqn:Lu_g_x} and  \eqref{eqn:norm_g_x}, 
we obtain convex program \eqref{eqn:opt_dc_gain_2} by replacing $\vcg$ in \eqref{eqn:opt_dc_gain_1} with its parametric form.  
This concludes the proof.	
\qed
\subsection{Proof of Corollary~\ref{cor:main_thm_SE}} \label{sec:appendix_proof_main_thm_SE}
The convex program \eqref{eqn:opt_dc_gain_3} is a special case of \eqref{eqn:opt_dc_gain_1} where 
$\Ical=\{1,\ldots,\nD\}$, 
function $\ell:\Rbb^{\nD}\times\Rbb^{\nD}\to\Rbb_+$
is defined as $\ell(\vcv_1,\vcv_2)=\|\vcv_1-\vcv_2\|^2$, for $\vcv_1,\vcv_2\in\Rbb^{\nD}$, and, function $\rho:\Rbb_+\to\Rbb_+$ is defined as $\rho(r)=r^2$, for $r\in\Rbb_+$.
Subsequently, the existence and the uniqueness of the solution, and also the parametric form \eqref{eqn:gstar_parametric_form} are provided by Theorem \ref{thm:main_thm}.
Moreover, for the given $\ell$ and $\rho$, the optimization problem \eqref{eqn:opt_dc_gain_2} is reformulated as the following quadratic program
\begin{equation}\label{eqn:opt_dc_gain_4}
	\begin{array}{cl}
		\minOp_{\vcx\in\Rbb^{\nD+1}} & \ \|\mxA\vcx-\vcy\|^2
		+
		\lambda \vcx^\tr\Phi\vcx\\
		\mathrm{s.t.} & \ \vca_0^\tr\vcx=\delta.
	\end{array}	
\end{equation}
One can see that \eqref{eqn:QP_KKT} is the first-order necessary optimality condition for \eqref{eqn:opt_dc_gain_4}, and $\gamma$ is the Lagrange multiplier corresponding to the steady-state gain constraint $\vca_0^\tr\vcx=\delta$. 
When $\varphi_0,\ldots,\varphi_{\nD}$ are linearly independent, the Gram matrix $\Phi$ is positive definite, and consequently, the objective function of \eqref{eqn:opt_dc_gain_4} is strongly convex. 
Therefore, \eqref{eqn:opt_dc_gain_4} has a unique solution $\xstar$.
By replacing $\gamma$ with $\delta-\vca_0^\tr\vcx$ in \eqref{eqn:QP_KKT} and applying matrix inversion lemma, one can solve linear system of equations \eqref{eqn:QP_KKT} and obtain $\xstar$  as in \eqref{eqn:xstar}.	
\qed
\subsection{Proof of Theorem~\ref{thm:Phi_DT}} \label{sec:appendix_proof_Phi_DT}
From Lemma \ref{thm:Lu_bounded}, the definition of matrices  $\mx{K}$ and $\mx{T}_{\vc{u}}$, and since $u_t=0$ for $t<0$, one can see that
\begin{equation}
	\begin{split}
		\inner{\varphi_i}{\varphi_j}
		&=
		\sum_{s=0}^{i-1}\varphi_{j,s}u_{i-1-s}
		\\&=
		\sum_{s=0}^{i-1}
		\sum_{t=0}^{j-1}
		\kernel(s,t)u_{j-1-t}
		u_{i-1-s}
		\\&=[\mx{T}_{\vc{u}}\mx{K}\mx{T}_{\vc{u}}^\tr
		]_{(i,j)},
	\end{split}
\end{equation}
for any $i,j\in\{1,\ldots,\nD\}$.
Moreover, due to \eqref{eqn:phi_0t}, we have
\begin{equation}
	\inner{\varphi_i}{\varphi_0}
	=
	\sum_{s=0}^{i-1}\varphi_{0,s}u_{i-1-s}
	=
	[\mx{T}_{\vc{u}}\varphi]_{(i)},
\end{equation}
for any $i\in\{1,\ldots,\nD\}$.
Following this, the claim concludes from the definition of matrix $\Phi$ and the fact that $\Phi=\Phi^\tr$.
\qed
\subsection{Steady-State Gain Representer for Discrete-Time Standard Stable Kernels}
\label{sec:phi_0t_norm_phi_0_kTC_kDC_kSS}
\begin{theorem}\label{thm:phi_0t_DT}
	\emph{i)} For kernel $\kernelTC$, we have
	\begin{equation}\label{eqn:phi_0t_kTC_DT}
		\varphi_{\TC,0,t} =  
			\left(t+\frac{1}{1-\alpha}\right)\alpha^t,			
			\quad\forall t\in\Zbb_+.
	\end{equation}
	\emph{ii)} For kernel $\kernelDC$, we have
	\begin{equation}\label{eqn:phi_0t_kDC_DT}
		\varphi_{\DC,0,t} =  
			\Big[\frac{\gamma-\gamma^{t+1}}{1-\gamma}
			+
			\frac{1}{1-\alpha\gamma}\Big]\alpha^{t},		
			\quad\forall t\in\Zbb_+.
	\end{equation}
	\emph{iii)} For kernel $\kernelSS$, we have
	\begin{equation}\label{eqn:phi_0t_kSS_DT}
		\varphi_{\SS,0,t} = 
			\Big[\frac{1+\alpha-\alpha^{t+1}}{1-\alpha^2}-\frac{\alpha^{t}}{3(1-\alpha^3)}-\frac{t\alpha^{t}}{3}\Big]\alpha^{2t}\!,		
			\quad\forall t\in\Zbb_+.
	\end{equation}
\end{theorem}
\begin{proof}
	i) From \eqref{eqn:phi_0t}, we have
	\begin{equation*}
		\begin{split}\!\!\!\!\!
			\varphi_{\TC,0,t} &
			=
			\sum_{s=0}^{t-1}\alpha^{\max(s,t)}
			+
			\sum_{s\ge t}\alpha^{\max(s,t)}
			\\&=
			t\alpha^t
			+
			\frac{\alpha^t}{1-\alpha}
			= \left(t+\frac{1}{1-\alpha}\right)\alpha^t.
		\end{split}
	\end{equation*}
	ii) Due to \eqref{eqn:phi_0t}, one has
	\begin{equation*}
		\begin{split}\!\!\!\!\!
			\varphi_{\DC,0,t} &
			=
			\sum_{s=0}^{t-1}\alpha^{\max(s,t)}\gamma^{|s-t|}
			+
			\sum_{s\ge t}\alpha^{\max(s,t)}\gamma^{|s-t|}
			\\&=
			\alpha^{t}\sum_{s=0}^{t-1}\gamma^{-s+t}
			+
			\alpha^t\sum_{s\ge t}\alpha^{s-t}\gamma^{s-t}
			\\&=
			\alpha^{t}\gamma\frac{1-\gamma^t}{1-\gamma}
			+
			\frac{\alpha^{t}}{1-\alpha\gamma}
			=
			\Big[\frac{\gamma-\gamma^{t+1}}{1-\gamma}
			+
			\frac{1}{1-\alpha\gamma}\Big]\alpha^{t}.
		\end{split}
	\end{equation*}
	iii) We know that
	\begin{equation*}
		\begin{split}\!\!\!\!\!
			\sum_{s=0}^{t-1}&\alpha^{\max(s,t)}\alpha^{s+t}
			+
			\sum_{s\ge t}\alpha^{\max(s,t)}\alpha^{s+t}
			\\&=
			\alpha^{2t}\sum_{s=0}^{t-1}\alpha^{s}
			+
			\alpha^t\sum_{s\ge t}\alpha^{2s}
			\\&=
			\alpha^{2t}\frac{1-\alpha^t}{1-\alpha}
			+
			\frac{\alpha^{3t}}{1-\alpha^2}
			=
			\Big[\frac{1+\alpha-\alpha^{t+1}}{1-\alpha^2}\Big]\alpha^{2t}.
		\end{split}
	\end{equation*}
	Following this, from \eqref{eqn:phi_0t} and $\varphi_{\TC,0,t}$, we have
	\begin{equation}
		\begin{split}
			\varphi_{\SS,0,t}&=
			\Big[\frac{1+\alpha-\alpha^{t+1}}{1-\alpha^2}\Big]\alpha^{2t}
			-\frac{1}{3}(t+\frac{1}{1-\alpha^3})\alpha^{3t}
			\\&=\Big[\frac{1+\alpha-\alpha^{t+1}}{1-\alpha^2}-\frac{\alpha^{t}}{3(1-\alpha^3)}-\frac13t\alpha^{t}\Big]\alpha^{2t}.
		\end{split}
	\end{equation}
This concludes the proof.
\end{proof}	
\begin{theorem}\label{thm:norm_phi_0_DT}
	\emph{i)} For kernel $\kernelTC$, we have
	\begin{equation}\label{eqn:norm_phi_0t_kTC_DT}
		\|\varphi_{\TC,0}\|_{\Hk}^2 =  
			\frac{\alpha+1}{(1-\alpha)^2}.
	\end{equation}
	\emph{ii)} For kernel $\kernelDC$, we have
	\begin{equation}
		\|\varphi_{\DC,0}\|_{\Hk}^2 =   
			\frac{1+\alpha\gamma}{(1-\alpha)(1-\alpha\gamma)}.
	\end{equation}
	\emph{iii)} For kernel $\kernelSS$, we have
	\begin{equation}
		\|\varphi_{\SS,0}\|_{\Hk}^2 = 
			\frac{2}{3}
			\frac{\alpha^4 + \alpha^3 + 3\alpha^2 + \alpha + 1}{(\alpha^3 - 1)^2(\alpha + 1)}.
	\end{equation}
\end{theorem}
\begin{proof}
	i) We know that $\|\varphi_{0}\|^2  = \sum_{t\in\Zbb_+}\varphi_{0,t}$. Hence,
	\begin{equation*}
		\begin{split}\!\!\!\!\!
			\|\varphi_{\TC,0}\|_{\Hk}^2 
			&=\sum_{t\in\Zbb_+}t\alpha^{t} + 
			\frac{1}{1-\alpha}\sum_{t\in\Zbb_+}\alpha^{t}
			\\&=
			\frac{\alpha}{(1-\alpha)^2}+\frac{1}{(1-\alpha)^2}=
			\frac{1+\alpha}{(1-\alpha)^2}.
		\end{split}
	\end{equation*}
	ii) Similarly to the proof of part i), for $\kernel=\kernelDC$, we have
	\begin{equation*}
		\begin{split}\!\!\!\!\!
			\|\varphi_{\DC,0}\|_{\Hk}^2 
			& = 
			\Big[\frac{\gamma}{1-\gamma}
			+
			\frac{1}{1-\alpha\gamma}\Big]\sum_{t\in\Zbb_+}\alpha^{t}
			- \frac{\gamma}{(1-\gamma)}
			\sum_{t\in\Zbb_+}(\alpha\gamma)^{t}
			\\=&
			\Big[\frac{\gamma}{1-\gamma} + \frac{1}{1-\alpha\gamma}\Big]\frac{1}{1-\alpha}
			- \frac{\gamma}{(1-\gamma)}\frac{1}{(1-\alpha\gamma)}
			\\=&
			\frac{1-\alpha\gamma^2-\gamma+\alpha\gamma}{(1-\alpha)(1-\gamma)(1-\alpha\gamma)}=
			\frac{1+\alpha\gamma}{(1-\alpha)(1-\alpha\gamma)}.
		\end{split}
	\end{equation*}
	iii) Similar to the previous parts, one has
	\begin{equation*}
		\begin{split}
			\|&\varphi_{\SS,0}\|_{\Hk}^2 
			\\&=
			\!\sum_{t\in\Zbb_+}\!\frac{\alpha^{2t}}{1-\alpha}-
			\Big[\frac{\alpha}{1-\alpha^2}
			+
			\frac{1}{3(1-\alpha^3)}\Big]
			\!\sum_{t\in\Zbb_+}\!
			\alpha^{3t}
			- 
			\!\sum_{t\in\Zbb_+}\!
			\frac{t\alpha^{3t}}{3}
			\\&=
			\frac{1}{(1-\alpha)(1-\alpha^2)}
			- 
			\frac{3\alpha-3\alpha^4+1-\alpha^2}{3(1-\alpha^2)(1-\alpha^3)^2}
			-\frac{\alpha^3}{3(1-\alpha^3)^2}
			\\&=\frac{2}{3}
			\frac{\alpha^4 + \alpha^3 + 3\alpha^2 + \alpha + 1}{(\alpha^3 - 1)^2(\alpha + 1)}.
		\end{split}
	\end{equation*}
This concludes the proof.
\end{proof}	
\subsection{Proof of Theorem~\ref{thm:phiu_tau_t_general}} \label{sec:appendix_proof_phiu_tau_t_general}
	From \eqref{eqn:phiu_tau_t} and \eqref{eqn:u_pw}, one can see that
\begin{equation}
	\begin{split}
		\phiu{\tau,t} 
		&= 
		\int_{\Rbb_+}\kernel(t,s)u_{\tau-s}\drm s
		\\
		&= 
		\int_{\Rbb_+}\kernel(t,s)\sum_{i=0}^{\nS-1} \xi_{i+1} \one_{[s_i,s_{i+1})}(\tau-s)\drm s
		\\&= 
		\sum_{i=0}^{\nS-1} \xi_{i+1} \int_{\bar{s}_{i+1}(\tau)}^{\bar{s}_{i}(\tau)}\kernel(t,s)\drm s.
	\end{split}
\end{equation}	
Following this, the claim is implied by the definition of function $\psi$ in \eqref{eqn:psi_def_general}. 
\subsection{Proof of Theorem~\ref{thm:psi_TC_DC_SS}} \label{sec:appendix_proof_psi_TC_DC_SS}
First, we consider the case of TC kernel. For $t\le a$, we have
\begin{equation}\label{eqn:psi_TC_t_le_a}
	\psiTC(t,a,b)
	= \int_a^b\alpha^s\drm s 
	= \frac{\alpha^b-\alpha^a}{\ln(\alpha)}.
\end{equation}	
Also, for $t\ge b$, one has
\begin{equation}\label{eqn:psi_TC_t_ge_b}
	\psiTC(t,a,b)
	= \int_a^b\alpha^t\drm s 
	= (b-a)\alpha^t.
\end{equation}	
Similar to the previous cases, when $t\in(a,b)$, we have
\begin{equation}
	\label{eqn:psi_TC_t_ge_a_le_b}
	\begin{split}
		\psiTC(t,a,b)
		&=
		\int_a^t\alpha^t\drm s
		+ 
		\int_t^b\alpha^s\drm s 
		\\&
		=
		(t-a)\alpha^t
		+
		\frac{\alpha^b-\alpha^t}{\ln(\alpha)}.
	\end{split}
\end{equation}	
Due to 
\eqref{eqn:psi_TC_t_le_a}, 
\eqref{eqn:psi_TC_t_ge_b}, 
\eqref{eqn:psi_TC_t_ge_a_le_b}, 
and the definition of $\eta$, one can see 
\eqref{eqn:psi_TC} holds. 
We have a similar argument for DC kernel.
More precisely, for $t\le a$, one can see that
\begin{equation}\label{eqn:psi_DC_t_le_a}
	\begin{split}
		\psiDC(t,a,b)
		&=
		\int_a^b \alpha^s\gamma^{s-t}\drm s 
		= \frac{(\alpha\gamma)^b-(\alpha\gamma)^a}{\ln(\alpha\gamma)}\gamma^{-t}.
	\end{split}
\end{equation}	
Also, for $t\ge b$, we have
\begin{equation}\label{eqn:psi_DC_t_ge_b}
	\begin{split}
		\psiDC(t,a,b)
		&=
		\int_a^b \alpha^t\gamma^{t-s}\drm s 
		= 
		\frac{\gamma^{-a}-\gamma^{-b}}{\ln(\gamma)}(\alpha\gamma)^t.
	\end{split}
\end{equation}	
Similarly, if $t\in(a,b)$, one has
\begin{equation}
	\label{eqn:psi_DC_t_ge_a_le_b}
	\begin{split}
		\psiDC(t,a,b)
		&=
		\int_a^t\alpha^t\gamma^{t-s}\drm s
		+ 
		\int_t^b\alpha^s\gamma^{s-t}\drm s 
		\\&
		=
		\frac{\gamma^{-a}-\gamma^{-t}}{\ln(\gamma)}(\alpha\gamma)^t
		+ 
		\frac{(\alpha\gamma)^b-(\alpha\gamma)^t}{\ln(\alpha\gamma)}\gamma^{-t}.
	\end{split}\!\!\!\!
\end{equation}	
From the definition of $\eta$, one can see that \eqref{eqn:psi_DC_t_le_a}, 
\eqref{eqn:psi_DC_t_ge_b} and
\eqref{eqn:psi_DC_t_ge_a_le_b} implies
\eqref{eqn:psi_TC}.
It remains to obtain the result for SS kernel. For $t\le a$, we have
\begin{equation}\label{eqn:psi_SS_t_le_a}
	\begin{split}
		\int_a^b& \alpha^{\max(s,t)+s+t}\drm s 
		=
		\int_a^b \alpha^{2s+t}\drm s 
		=
		\frac{\alpha^{2b}-\alpha^{2a}}{2\ln(\alpha)}\alpha^{t}.
	\end{split}
\end{equation}	
Also, if $t\ge b$, one can see that
\begin{equation}\label{eqn:psi_SS_t_ge_b}
	\begin{split}
		\int_a^b& \alpha^{\max(s,t)+s+t}\drm s 
		=
		\int_a^b \alpha^{2t+s}\drm s 
		=		\frac{\alpha^{b}-\alpha^{a}}{\ln(\alpha)}\alpha^{2t}.
	\end{split}
\end{equation}	
Similarly, if $t\in(a,b)$, one has
\begin{equation}
	\label{eqn:psi_SS_t_ge_a_le_b}
	\begin{split}
		\int_a^b \alpha^{\max(s,t)+s+t}\drm s 
		&=
		\int_a^t\alpha^{2t+s}\drm s
		+ 
		\int_t^b\alpha^{2s+t}\drm s 
		\\&
		=
		\frac{\alpha^{t}-\alpha^{a}}{\ln(\alpha)}\alpha^{2t}
		+ 
		\frac{\alpha^{2b}-\alpha^{2t}}{2\ln(\alpha)}\alpha^{t}.
	\end{split}
\end{equation}	
Based on the definition of $\eta$,  \eqref{eqn:psi_SS_t_le_a}, 
\eqref{eqn:psi_SS_t_ge_b}, and \eqref{eqn:psi_SS_t_ge_a_le_b},
one can see that
\begin{equation}
	\begin{split}
	\int_a^b &\alpha^{\max(s,t)+s+t}\drm s \\
	&=
	\frac{\alpha^{\eta(t,a,b)}-\alpha^{a}}{\ln(\alpha)}
	\alpha^{2t}
	+
	\frac{\alpha^{2b}-\alpha^{2\eta(t,a,b)}}{2\ln(\alpha)}
	\alpha^{t}.
	\end{split}
\end{equation}	
Accordingly, by replacing $\alpha$ with $\alpha^3$ is \eqref{eqn:psi_TC} and due to the definition of $\kernelSS$, 
one can obtain \eqref{eqn:psi_SS}.
\qed
\subsection{Proof of Theorem~\ref{thm:phi_0t_CT}} \label{sec:appendix_proof_phi_0t_CT}
From the definition of $\varphi_0$ and function $\psi$, we have
\begin{equation*}
	\varphi_{0,t}
	=
	\int_{0}^{\infty}\!\kernel(s,t)\drm s 
	= 
	\lim_{b\to\infty}\int_{0}^{b}\!\kernel(s,t)\drm s 
	=	
	\lim_{b\to\infty}\psi(t,0,b).	
\end{equation*}
Subsequently, one can see that the theorem holds according to Theorem \ref{thm:psi_TC_DC_SS}.	
\qed
\subsection{Proof of Theorem~\ref{thm:inner_phi_0i_ij_CT}} \label{sec:appendix_proof_inner_phi_0i_ij_CT}
Due to Theorem \ref{thm:dc_bounded}, Lemma \ref{thm:Lu_bounded} and \eqref{eqn:u_pw}, one has
\begin{equation}
	\label{eqn:inner_phi_0_phiu_tau_1}
	\begin{split}
		&\inner{\varphi_0}{\phiu{\tau}}_{\Hk}
		=
		\int_{0}^{\infty}\varphi_{0,s}
		\sum_{i=0}^{\nS-1} \xi_{i+1} \one_{[s_i,s_{i+1})}(\tau-s)\drm s 
		\\&=
		\int_{0}^{\infty}\int_{0}^{\infty}
		\kernel(s,t)
		\sum_{i=0}^{\nS-1} \xi_{i+1} \one_{[s_i,s_{i+1})}(\tau-s)\drm t  \drm s.
	\end{split}	
\end{equation}
Accordingly, from the definition of function $\barnu$ and functions $\bar{\kappa}_i$, $i=0,\ldots,\nS-1$, we have
\begin{equation}
	\label{eqn:inner_phi_0_phiu_tau_2}
	\begin{split}
		&\!\!\!
		\inner{\varphi_0}{\phiu{\tau}}_{\Hk}
		=
		\sum_{i=0}^{\nS-1} 
		\xi_{i+1}
		\int_{\bar{s}_{i+1}(\tau)}^{\bar{s}_{i}(\tau)}
		\int_{0}^{\infty}
		\kernel(s,t)\drm t \ \! \drm s
		\\&=
		\sum_{i=0}^{\nS-1} 
		\xi_{i+1}
		\big(\barnu(\bar{s}_{i}(\tau)) 
		-
		\barnu(\bar{s}_{i+1}(\tau)\big)
		\\&= 
		\sum_{i=0}^{\nS-1}
		\xi_{i+1} \bar{\kappa}_{i}(\tau).
	\end{split}	
\end{equation}
From Lemma \ref{thm:Lu_bounded} and due to \eqref{eqn:phiu_tau_t} and \eqref{eqn:u_pw}, we know that
\begin{equation}
	\label{eqn:inner_phiu_tau_12}
	\begin{split}
		&\inner{\phiu{\tau_1}}{\phiu{\tau_2}}_{\Hk}
		=
		\int_{0}^{\infty}\phiu{\tau_2,s}
		\sum_{i=0}^{\nS-1} \xi_{i+1} \one_{[s_i,s_{i+1})}(\tau_1-s)\drm s 
		\\&=
		\int_{0}^{\infty}\int_{0}^{\infty}\bigg(
		\kernel(s,t)
		\sum_{j=0}^{\nS-1} \xi_{j+1} \one_{[s_j,s_{j+1})}(\tau_2-t)
		\\&\qquad\qquad\qquad
		\sum_{i=0}^{\nS-1} \xi_{i+1} \one_{[s_i,s_{i+1})}(\tau_1-s)\bigg)
		\drm t \ \! \drm s
		\\&=
		\sum_{i=0}^{\nS-1}\sum_{j=0}^{\nS-1} 
		\xi_{i+1}\xi_{j+1}
		\int_{0}^{\infty}\int_{0}^{\infty}\bigg(
		\kernel(s,t)
		\\&\qquad\qquad\qquad
		\one_{[s_j,s_{j+1})}(\tau_2-t)
		\one_{[s_i,s_{i+1})}(\tau_1-s)\bigg)
		\drm t \ \! \drm s.
	\end{split}	\!\!\!\!\!\!\!\!\!\!\!
\end{equation}
From \eqref{eqn:nu_def_general}, it follows that
\begin{equation}\label{eqn:intintk11}
	\begin{split}
		&\!\!\!\int_{0}^{\infty}\!\!\!\int_{0}^{\infty}\!\!  
		\kernel(s,t)
		\one_{[s_j,s_{j+1})}(\tau_2-t)
		\one_{[s_i,s_{i+1})}(\tau_1-s)\drm t \ \! \drm s 	
		\\&=
		\int_{\bar{s}_{i+1}(\tau_1)}^{\bar{s}_{i}(\tau_1)}\!
		\int_{\bar{s}_{j+1}(\tau_2)}^{\bar{s}_{j}(\tau_2)}\! 
		\kernel(s,t)
		\drm t \ \! \drm s
		\\&=
		\nu\big(\bar{s}_{i}(\tau_1),\bar{s}_{j}(\tau_2)\big) -
		\nu\big(\bar{s}_{i+1}(\tau_1),\bar{s}_{j}(\tau_2)\big)
		\\&\qquad
		-\nu\big(\bar{s}_{i}(\tau_1),\bar{s}_{j+1}(\tau_2)\big) +
		\nu\big(\bar{s}_{i+1}(\tau_1),\bar{s}_{j+1}(\tau_2)\big).
	\end{split}
\end{equation}
Thus, the claim is implied from \eqref{eqn:inner_phiu_tau_12} and \eqref{eqn:intintk11}. 
\qed
\subsection{Proof of Theorem~\ref{thm:nu_TC_DC_SS}} \label{sec:appendix_proof_nu_TC_DC_SS}
Since $\barnu(x)=\lim_{y\to\infty}\nu(x,y)$, it is enough to obtain $\nu(x,y)$. Throughout the proof, without loss of generality we assume $x\le y$. Due to the definition of $\nu$ in \eqref{eqn:nu_def_general}, for TC kernel, we have
\begin{equation}\label{eqn:nu_TC_pf_01}
	\begin{split}
		\nuTC(x,y)
		&= 
		\int_0^x\int_0^s\alpha^s\drm t\drm s 
		+
		\int_0^x\int_s^y\alpha^t\drm t\drm s 
		\\&= 
		\int_0^x s\alpha^s\drm s 
		+
		\int_0^x\frac{\alpha^y-\alpha^s}{\ln(\alpha)}\drm s
		\\&=
		\frac{x\alpha^x\ln(\alpha)+1-\alpha^x}{\ln(\alpha)^2}
		+ 
		\frac{x\alpha^y}{\ln(\alpha)}
		- 
		\frac{\alpha^x-1}{\ln(\alpha)^2}. 
	\end{split}
\end{equation}	
Reordering the terms and replacing $x$ and $y$ respectively with $\min(x,y)$ and $\max(x,y)$, we obtain $\nuTC(x,y)$.
Letting $y$ go to $\infty$ in  $\nuTC(x,y)$, we obtain  $\barnuTC(x)$. Similarly, from the definition of $\nu$, we have
\begin{equation}\label{eqn:nu_DC_pf_01}
	\begin{split}
		&\!\!\nuDC(x,y)
		\\&= 
		\int_0^x\int_0^s\alpha^s\gamma^{s-t}\drm t\drm s 
		+
		\int_0^x\int_s^y\alpha^t\gamma^{t-s}\drm t\drm s 
		\\&= 
		\int_0^x\alpha^s\gamma^{s}
		\int_0^s\gamma^{-t}\drm t\drm s 
		+
		\int_0^x \gamma^{-s}
		\int_s^y \alpha^t\gamma^{t}\drm t\drm s 
		\\&= 
		\int_0^x (\alpha\gamma)^{s}
		\frac{1-\gamma^{-s}}{\ln(\gamma)}
		\drm s 
		+
		\int_0^x \gamma^{-s}
		\frac{(\alpha\gamma)^y-(\alpha\gamma)^{s}}{\ln(\alpha\gamma)}
		\drm s 
		\\&= 
		\frac{(\alpha\gamma)^x-1}
		{\ln(\gamma)\ln(\alpha\gamma)}
		-
		\frac{\alpha^x-1}
		{\ln(\gamma)\ln(\alpha)}
		-
		\frac{(\alpha\gamma)^y(\gamma^{-x}-1)}
		{\ln(\gamma)\ln(\alpha\gamma)}
		\\&\qquad -
		\frac{\alpha^{x}-1}
		{\ln(\alpha)\ln(\alpha\gamma)}.
	\end{split}
\end{equation}	
Similar to the previous part, by replacing $x$ and $y$ respectively with $\min(x,y)$ and $\max(x,y)$, and reordering the terms, we obtain $\nuDC(x,y)$.
Moreover, one can see that $\barnuTC(x)$ results when $y$ goes to $\infty$.  
Replacing $\alpha$ with $\alpha^3$ in $\nuTC(x,y)$, we obtain
\begin{equation}\label{eqn:nu_SS_pf_01}
	\begin{split}\!\!\!\!\!\!\!\!\!\!
		\int_0^x& \!\!\!\! \int_0^y  \!\!
		\alpha^{3\max(s,t)}
		\drm t \drm s 
		\\&
		= 	\frac{\min(x,y)(\alpha^{3x}+\alpha^{3y})}{3\ln(\alpha)}
		+	
		\frac{2(1-\alpha^{3\min(x,y)})}
		{9\ln(\alpha)^2}.
	\end{split}
\end{equation}
On the other hand, we have
\begin{equation}\label{eqn:nu_SS_pf_02}
	\begin{split}
		\int_0^x\int_0^y &  
		\alpha^{\max(s,t)+s+t}
		\drm t \drm s 
		\\&= 
		\int_0^x\int_0^s 
		\alpha^{2s+t}
		\drm t \drm s 
		+
		\int_0^x \int_s^y 
		\alpha^{s+2t}
		\drm t \drm s 
		\\&
		= 
		\int_0^x \alpha^{2s}
		\frac{\alpha^s-1}{\ln(\alpha)}
		+ 
		\alpha^{s}
		\frac{\alpha^{2y}-\alpha^{2s}}{2\ln(\alpha)}
		\drm s 
		\\&
		= 
		\frac{\alpha^{3x}-1}{6\ln(\alpha)^2}
		-
		\frac{\alpha^{2x}-1}{2\ln(\alpha)^2}
		+ 
		\frac{\alpha^{2y}(\alpha^{x}-1)}{2\ln(\alpha)^2}. 
	\end{split}
\end{equation}	
Using the definition of $\nu$, \eqref{eqn:nu_SS_pf_01}, \eqref{eqn:nu_SS_pf_02},
replacing $x$ and $y$ respectively with $\min(x,y)$ and $\max(x,y)$, and reordering the terms, we obtain $\nuSS(x,y)$.
Also, letting $y\to\infty$, we get $\barnuSS(x)$. 
\qed
\subsection{Proof of Theorem~\ref{thm:norm_phi_0_CT}} \label{sec:appendix_proof_norm_phi_0_CT}
	From \eqref{eqn:norm_phi_0} and the definition of function $\nu$, we have
\begin{equation*}
	\begin{split}
		\|\varphi_{0}\|_{\Hk}^2
		&=
		\int_{0}^{\infty}\!\int_{0}^{\infty}\!\kernel(s,t)\drm s \drm t 
		\\&= 
		\lim_{x,y\to\infty}
		\int_{0}^{x}\!\int_{0}^{y}\!\kernel(s,t)\drm s \drm t =	
		\lim_{x,y\to\infty}\nu(x,y).	
	\end{split}
\end{equation*}
Accordingly, the proof of the theorem follows directly from Theorem \ref{thm:nu_TC_DC_SS}.	
\qed
\subsection{Representer Theorem} \label{sec:appendix_rep_thm}
There are several variations of the representer theorem in the literature, which vary primarily in terms of generality. The following version is borrowed from \cite{dinuzzo2012representer}.  
\begin{theorem}[\cite{dinuzzo2012representer}]\label{thm:rep_thm}
	Let $\Hilbert$ be a Hilbert space endowed with inner product $\inner{\cdot}{\cdot}_{\Hilbert}$ and $r:\Rbb_+\to \Rbb$ be an increasing function. Consider the following optimization problem 
	\begin{equation}
		\label{eqn:opt_representer_thm}
		\min_{\vc{w}\in\Hilbert} \    
		e(\inner{\vc{w}_1}{\vc{w}}_{\Hilbert}, \ldots,
		\inner{\vc{w}_m}{\vc{w}}_{\Hilbert})+ r\big(\big\|\vc{w}\big\|_{\Hilbert}\big),
	\end{equation}
	where  $\vc{w}_1,\ldots,\vc{w}_m$ are vectors in $\Hilbert$ and $e:\Rbb^m\to \Rbb\cup\{+\infty\}$ is a given function. 
	Then, \eqref{eqn:opt_representer_thm} has a solution in $\Wscr:=\linspan\{\vc{w}_i\}_{i=1}^m$, when it admits a solution.
\end{theorem}
It is worth noting that Theorem~\ref{thm:rep_thm} is a generalized form of the representer theorem discussed in 
\cite[Theorem 1.3.1]{wahba1990spline}.

\bibliographystyle{IEEEtran}
\bibliography{mainbib}

\vfill
\end{document}